\documentclass[journal]{IEEEtran}
\usepackage[T1]{fontenc}

\ifCLASSINFOpdf
\else
   \usepackage[dvips]{graphicx}
\fi
\usepackage{url}
\usepackage{graphicx}
\usepackage{stfloats}
\usepackage{placeins}

\usepackage{enumitem}
\usepackage{xcolor} 

\usepackage{tikz}
\usetikzlibrary{
  plotmarks,
  arrows.meta,
  positioning,
  calc,
  patterns
}
\usepackage{bbm}
\usepackage{pgfplots}
\pgfplotsset{
  compat=1.18,
  releasecount xaxis/.style={
    xlabel={Number of releases $L$},
    xlabel near ticks,
  },
  unlab crb yaxis/.style={
    ylabel={$\mathrm{CRB}_{\mathrm{unlab}}(L)$},
    ymode=log,
    ylabel near ticks,
  },
  unlab crb yaxis fixed Lid/.style={
    ylabel={$\mathrm{CRB}_{\mathrm{unlab}}(L_{\mathrm{id}})$},
    ymode=log,
    ylabel near ticks,
  },
  labeled mse yaxis/.style={
    ylabel={$\mathrm{MSE}(x_j)$},
    ylabel near ticks,
  },
}
\usepgfplotslibrary{
  groupplots,
  dateplot,
  statistics,
  fillbetween,
  patchplots
}

\usepackage{pgfplotstable} 
\usepackage{siunitx}
\usepackage{booktabs}
\usepackage{bm}
\usepackage{cite}
\usepackage{amsmath,amssymb,amsfonts}
\usepackage{textcomp}

\usepackage{xspace}

\def\bG{{\bm{G}}}
\def\bg{{\bm{g}}}
\def\bJ{{\bm{J}}}

\def\bB{{\bm{B}}}

\def\bu{{\bm{u}}}

\usepackage{tikz}
\usepackage{algorithm}
\usepackage[noend]{algpseudocode}
\usepackage[normalem]{ulem}

\usepackage{amsmath,amsfonts,bm}
\usepackage{amssymb}
\DeclareMathOperator{\range}{range}

\usepackage{amsthm}
\theoremstyle{plain}
\newtheorem{remark}{Remark}
\newtheorem{lemma}{Lemma}

\newtheorem{proposition}{Proposition}
\newtheorem{corollary}{Corollary}

\newtheorem{theorem}{Theorem}

\def\1{\bm{1}}

\DeclareMathAlphabet{\mathsfit}{\encodingdefault}{\sfdefault}{m}{sl}
\SetMathAlphabet{\mathsfit}{bold}{\encodingdefault}{\sfdefault}{bx}{n}

\def\bq{{\bm{q}}}

\def\bx{{\bm{x}}}
\def\bw{{\bm{w}}}

\def\be{{\bm{e}}}

\def\by{{\bm{y}}}
\def\btheta{{\bm{\theta}}}
\def\bP{{\bm{P}}}
\def\bM{{\bm{M}}}

\newcommand{\R}{\mathbb{R}}

\newcommand{\Var}{\mathrm{Var}}

\newcommand{\tr}{\mathrm{tr}}

\DeclareMathOperator{\rank}{rank}

\newcommand{\bI}{\bm I}

\newcommand{\bSigma}{\bm{\Sigma}}

\newcommand{\bC}{\bm C}

\newcommand{\bA}{\bm A}

\newcommand{\bp}{\bm p}

\newcommand{\bW}{\bm W}

\usepackage[colorlinks=true,
            linkcolor=blue,
            citecolor=blue,
            urlcolor=blue]{hyperref}
\hypersetup{
  colorlinks=true,
  linkcolor=blue,
  citecolor=blue,
  urlcolor=blue
}
\newcommand{\appref}[2]{Appendix~\hyperref[#2]{\ref*{#1}}}

\begin{document}

\title{Reconstruction Limits for Repeated Differentially Private Aggregates:\\
A Cram\'er--Rao Perspective on Query Geometry}
\author{Chenyue Zhang, \IEEEmembership{Member, IEEE}, Andrew Campbell, \IEEEmembership{Student Member, IEEE}, Anna Scaglione, \IEEEmembership{Fellow, IEEE}, and Sean Peisert, \IEEEmembership{Senior Member, IEEE}%

\thanks{This work was supported by the U.S. Department of Energy, Office of Cybersecurity, Energy Security, and Emergency Response (CESER), through the Privacy-Preserving, Collective Cyberattack Defense of DERs Project under contract \mbox{DE-AC02-05CH11231}.}
\thanks{Chenyue Zhang, Andrew Campbell, and Anna Scaglione are with the Department of Electrical and Computer Engineering, Cornell University, Cornell Tech, New York, NY 10044, USA (e-mail: cz563@cornell.edu; ac2458@cornell.edu; as337@cornell.edu).}
\thanks{Sean Peisert is with Lawrence Berkeley National Laboratory, Berkeley, CA, USA (e-mail: sppeisert@lbl.gov).}%

}



\maketitle

\begin{abstract}
Repeated differentially private (DP) releases are often evaluated by transcript
length or cumulative privacy accounting. We show that these quantities do not
by themselves determine local reconstruction risk. For Gaussian-calibrated
repeated statistical queries, the key object is the nuisance-profiled Fisher
geometry of the release sequence: repetition helps only when new releases
create identifiable directions after nuisance variables are removed. Thus,
release geometry determines what can be locally identified, while the privacy
accountant determines how precisely those directions can be estimated.

We develop this principle in two settings. For labeled-target reconstruction
with fixed-background \textsf{IN/OUT} averages, repeated copies collapse to a
single target-versus-background contrast. The best linear unbiased estimator
attains the Cram\'er--Rao bound, and additional copies provide only averaging
gain; under Basic Composition this gain is dominated by the
$\Theta(L^2\log L)$ noise penalty, whereas zCDP/RDP-style Gaussian accounting
makes the risk order-flat. For static permutation-invariant releases, labels
remain unidentified, but feature diversity can make the sorted participating
multiset locally identifiable. For polynomial moments and smooth thresholds,
the useful number of releases is governed by the balance between newly exposed
eigendirections and accountant-induced noise growth. These results provide a
local, mechanism-specific benchmark for value leakage in repeated private
sensing and analytics.
\end{abstract}
\begin{IEEEkeywords}
Differential privacy, reconstruction attacks, aggregate queries,
privacy accounting, Fisher information, Cram\'er--Rao bounds
\end{IEEEkeywords}

\section{Introduction}
\label{sec:introduction}
Repeated differentially private (DP) aggregate releases appear in systems that report statistics over the same or overlapping populations over time, including federated-learning diagnostics, analytics dashboards, and continual-observation reporting~\cite{McMahanRamageTalwarZhang2018,Dwork2010Continual}. Under a fixed total privacy budget, an added release is not automatically a new reconstruction equation: it may only spend budget on another noisy copy of the same profiled contrast. The release count alone is therefore an unreliable proxy for reconstruction risk; what matters is which data directions the release sequence makes identifiable.

Existing analyses of repeated differentially private releases typically focus
either on composition accounting or on transcript accuracy. Our contribution
is orthogonal: for a fixed release family, we characterize which data
directions become locally identifiable after nuisance profiling, independently
of how privacy accounting scales the noise. This separates the geometric
effect of adding releases from the accountant-dependent scaling of
reconstruction precision. The resulting principle is that repeated releases
improve reconstruction only when they create new profiled identifiable
directions.

This distinction matters for numerical reconstruction. DP controls how an output distribution changes under neighboring datasets~\cite{dwork2006calibrating,dwork2014algorithmic,dong2022gaussian}, and composition theorems describe how privacy budgets accumulate across releases~\cite{balle2018improving,mironov2017renyi,KairouzOhViswanath2017}. Once a transcript has been observed, however, a separate estimation question remains: how accurately can an adversary infer a sensitive record value~\cite{WassermanZhou2010}? This issue arises in settings where the record values themselves encode sensitive behavior or state: mobility traces can identify individuals, electricity load profiles can reveal household routines or occupancy, and physiological or wearable signals can carry biometric and health information~\cite{deMontjoye2013,Asghar2017SmartMeterPrivacy,Chen2013Occupancy,Odinaka2012ECG}. Existing attacks and bounds similarly show that limiting membership leakage does not necessarily prevent accurate value reconstruction~\cite{TramerBoneh2020,Guo2022Bounding,zhu2019deep,geiping2020inverting,hayes2023bounding}.

\paragraph*{Scope of the analysis}
The Fisher-information and Cram\'er--Rao quantities studied in this paper are
local, mechanism-specific reconstruction benchmarks. They characterize the
identifiable directions induced by a fixed Gaussian-calibrated release family
after nuisance profiling. They are not global differential-privacy guarantees,
minimax risks, Bayesian risks, or guarantees for arbitrary adversaries.
Differential privacy controls distinguishability between neighboring datasets,
whereas our analysis studies local reconstruction accuracy under a specified
transcript model and operating point. Its role is to separate the design
question, ``what directions does the transcript expose?'', from the accounting
question, ``how much precision is assigned to those directions?''.

The central message is no new geometry, no new target direction. If both the membership set and the per-record feature map remain fixed, repeated permutation-invariant releases resample the same profiled observation direction. Averaging can reduce noise along that direction, but it cannot enlarge the target-distinguishing span. New geometry enters only through \emph{membership variation}, where different subsets create asymmetry for a labeled target, or through \emph{feature diversity}, where fixed participants are queried through changing feature maps. In the latter case, permutation invariance still prevents label recovery, so away from collisions the locally recoverable object is the sorted participating multiset.

This paper makes three contributions. First, it gives a nuisance-profiled Fisher framework that separates release geometry from accounting. Second, it solves the labeled-target fixed-background \textsf{IN/OUT} case: all copies reduce to one contrast, with risk $\kappa_L/L$. Third, it analyzes static permutation-invariant releases, where feature diversity can restore rank only for the sorted multiset and only while the weakest identifiable directions outpace $\kappa_L$. For design, the relevant quantity is not transcript length, but the number and strength of target-identifiable directions that survive nuisance profiling after privacy calibration.

The remainder of the paper is organized as follows. Section~\ref{sec:framework} introduces the observation model and incremental Fisher geometry. Section~\ref{sec:lever1} studies labeled-target reconstruction under subset asymmetry. Section~\ref{sec:lever2} studies unlabeled multiset reconstruction under feature diversity. Section~\ref{sec:experiments} reports numerical experiments, and Section~\ref{sec:conclusion} concludes.

\subsection{Related Work}
\label{subsec:related}

Classical reconstruction results show that sufficiently accurate aggregate answers can reveal substantial information about the underlying database~\cite{DinurNissim2003,dwork2007price,dwork2014algorithmic}, with related risks also observed in empirical census studies~\cite{Abowd2023census,Dick2023confidence}. Those works typically treat answer accuracy as fixed or externally specified. Here the release noise is coupled to the number of releases by a fixed DP budget and a specified privacy accountant.

Privacy accounting provides a second point of departure. Composition theorems for approximate DP, R\'enyi DP, zero-concentrated DP, and Gaussian DP translate a total privacy budget into per-release Gaussian noise in different ways~\cite{dwork2014algorithmic,KairouzOhViswanath2017,mironov2017renyi,BunSteinke2016,dong2022gaussian}. We do not build the theory around one accountant. Each accountant enters through the scalar penalty $\kappa_L$; the reconstruction question is how that penalty interacts with the local rank and conditioning of the transcript.

Recent work has emphasized that comparable DP guarantees need not imply comparable numerical reconstruction risk~\cite{Guo2022Bounding,hayes2023bounding}. Much of that literature studies trained-model release, gradient leakage, or learning pipelines. Our object is a repeated aggregate transcript, so the geometry of the query sequence becomes central. The analysis is also related to private workload design, the matrix mechanism, Fisher-information privacy design, and optimal experimental design~\cite{Li2015matrix,Xiao2021fitness,Farokhi2018,Pukelsheim2006,ChalonerVerdinelli1995}. Those lines typically choose a workload, noise covariance, or experiment to optimize an error criterion. Here the release design is fixed; we characterize the reconstruction consequences once the accountant fixes the calibration.

The fixed-membership case further connects to unlabeled sensing and permutation recovery~\cite{Unnikrishnan2018,Haghighatshoar2017,ZhangSlawskiLi2020}. Classical unlabeled sensing asks how measurement diversity resolves the permutation ambiguity. The DP setting adds an accounting constraint: even when feature diversity restores rank in the quotient space, each additional measurement also changes $\kappa_L$, so local identifiability and the optimal release count can separate.

\section{Estimation-Theoretic Framework}
\label{sec:framework}
\subsection{Observation Model}
\label{sec:obs_model}

We first formalize the local experiment induced by a repeated DP aggregate
transcript. An adversary observes all releases and estimates a parameter of
interest; all other degrees of freedom are treated as nuisance parameters.
To cover both scalar and signal-valued records, write
$\bx_k\in\R^p$, $k=1,\ldots,N$, and collect the dataset as
\[
\bx_{1:N}:=(\bx_1,\ldots,\bx_N)\in\R^{Np}.
\]
The scalar case is recovered by $p=1$.\footnote{For linear
coordinate-wise aggregates, the passage from $p=1$ to general $p$
amounts to a Kronecker lifting; \appref{app:framework}{app:block_jacobians}
records the corresponding block-Jacobian and lift formulas.}

Stacking the $L$ released answers as $\by\in\R^L$, we model the
observation as
\begin{equation}
\label{eq:obs_model}
\by
=
\bq(\bx_{1:N})+\bw,
\qquad
\bw\sim\mathcal N(\bm 0,\bSigma),
\quad
\bSigma\succ \bm 0,
\end{equation}
where $\bq:\R^{Np}\to\R^L$ is the stacked deterministic query map
and $\bSigma$ is the covariance induced by the privacy mechanism and its
associated accounting rule.\footnote{Throughout the Fisher-information
analysis, we restrict attention to release families for which $\bq$ is
continuously differentiable in a neighborhood of the operating point, so
that the Jacobian with respect to the underlying continuous parameters is
well defined.}

The adversary is assumed to know the query map $\bq$ and, when
applicable, the membership sets $\{C_\ell\}_{\ell=1}^L$. For the
explicit DP examples studied below, we use additive Gaussian
calibration and make the covariance $\bSigma$ explicit. Let
$\Delta_{2,\ell}$ denote the $\ell_2$-sensitivity of the query
$q_\ell$ under the relevant adjacency relation. For release $\ell$, we
adopt the classical sufficient Gaussian calibration
\begin{equation}
\label{eq:gaussian_calibration}
w_\ell\sim\mathcal N(0,\sigma_\ell^2),
\qquad
\sigma_\ell^2
=
2\ln(1.25/\delta_\ell)\,\Delta_{2,\ell}^2/\varepsilon_\ell^2.
\end{equation}
The sensitivity-normalized per-release precision satisfies
$\mathrm{SNR}_\ell:=\Delta_{2,\ell}^2/\sigma_\ell^2
=\varepsilon_\ell^2/(2\ln(1.25/\delta_\ell))$, so the privacy budget
controls SNR through $\varepsilon_\ell^2$.
We use this classical sufficient calibration because it gives an explicit
scalar noise factor. For larger per-release budgets, the same geometric
analysis applies with the analytic Gaussian calibration; only the
accountant-dependent scalar $\kappa_L$ changes.
Under \emph{uniform Basic Composition}, the per-release budgets satisfy
$\varepsilon_\ell=\varepsilon_{\mathrm{tot}}/L$ and
$\delta_\ell=\delta_{\mathrm{tot}}/L$. Substituting these values
into~\eqref{eq:gaussian_calibration} gives
\begin{equation}
\label{eq:kappaL_sigma_def}
\begin{aligned}
\sigma_\ell^2&=\kappa_L\,\Delta_{2,\ell}^2,\\
\kappa_L
&:=
2L^2\ln(1.25L/\delta_{\mathrm{tot}})/\varepsilon_{\mathrm{tot}}^2
=
\Theta(L^2\log L).
\end{aligned}
\end{equation}
Assuming independent Gaussian perturbations across releases, the stacked
observation noise in~\eqref{eq:obs_model} has diagonal
covariance
\begin{equation}
\label{eq:Sigma_explicit}
\bSigma
=
\kappa_L\,\mathrm{diag}\!\bigl(\Delta_{2,1}^2,\ldots,\Delta_{2,L}^2\bigr).
\end{equation}

Equation~\eqref{eq:kappaL_sigma_def} is one plug-in, not the theory itself.
Throughout, the Fisher identities are geometry statements: additional rows
matter only if they add identifiable directions after nuisance projection.
The accountant is inserted by choosing the scalar growth law for $\kappa_L$.
Uniform Basic Composition gives $\kappa_L=\Theta(L^2\log L)$; uniform
Gaussian zCDP/RDP-style scalar accounting gives $\kappa_L=\Theta(L)$.

The Fisher-geometry results below are stated for arbitrary
positive-definite $\bSigma$; Sections~\ref{sec:lever1} and~\ref{sec:lever2}
apply the setup to concrete repeated-release designs under the
Gaussian calibration above.\footnote{Independent Laplace location noise
admits the same Fisher geometry; see \appref{app:framework}{app:laplace_noise}.}

\subsection*{Adjacency and Sensitivity Conventions}
\label{sec:adjacency}

The release families in Sections~\ref{sec:lever1}--\ref{sec:lever2} use different adjacency conventions. Table~\ref{tab:adjacency} records these conventions, so the asymptotic statements can be read on a common sensitivity scale.

\begin{table}[t]
\centering
\begingroup
\scriptsize
\setlength{\tabcolsep}{2pt}
\renewcommand{\arraystretch}{1.18}
\begin{tabular}{@{}p{0.18\linewidth}p{0.16\linewidth}p{0.20\linewidth}p{0.13\linewidth}p{0.12\linewidth}@{}}
\toprule
Family & Domain & Adjacency & $\Delta_{2,\ell}$ & Regularity \\
\midrule
\textsf{IN/OUT} (§\ref{sec:lever1}) & $x\in\mathbb{R}^N$, $\|x\|_\infty\le\Delta/2$ & record-level radius $\Delta$ & $\Delta/|\mathcal{C}_\ell|$ & linear \\
Polynomial moments (§\ref{subsec:lever2_poly}) & $|z_k|\le R$ & fixed-cardinality swap & $\le 2R^\ell/m$ & $C^\infty$ \\
Threshold queries (§\ref{subsec:lever2_threshold}) & $|z_k|\le R$ & replacement/swap & $\le 2/m$ & $C^1$ logistic/bump \\
\bottomrule
\end{tabular}
\endgroup
\caption{Adjacency conventions used in the release families. The asymptotic rates depend on these choices only through scalar constants; switching to add/remove adjacency multiplies the relevant $\Delta_{2,\ell}$ by a uniform factor.}
\label{tab:adjacency}
\end{table}

The general framework below specializes into two recurring cases used
throughout the paper. Section~\ref{sec:lever1} studies labeled-target
reconstruction under membership variation, where repeated releases may
distinguish one target record from a nuisance background through changing
subset geometry. Section~\ref{sec:lever2} studies unlabeled multiset
reconstruction under static membership, where permutation invariance prevents
label recovery and feature diversity determines which sorted-coordinate
directions become identifiable. In both settings, the central question is
whether additional releases create new identifiable directions after nuisance
projection.

\subsection{Fisher Geometry and the Nuisance-Aware Benchmark}
\label{subsec:fim_crb}

The analysis is local and parameter-of-interest based. We do not fix
a scalar target at the outset. Instead, we use a smooth local reparameterization
\[
\bx_{1:N}\longmapsto (\theta,\eta),
\]
where $\theta=T(\bx_{1:N})\in\R^r$ is the parameter of interest and
$\eta\in\R^{Np-r}$ collects the remaining nuisance coordinates. The same
formulation covers both settings studied later: labeled-target
reconstruction in Section~\ref{sec:lever1} and unlabeled multiset
reconstruction in Section~\ref{sec:lever2}. In the latter case, the
parameter of interest is represented locally by sorted coordinates away
from collisions.

Let
\begin{equation}
\label{eq:J_def}
\bJ
:=
\frac{\partial \bq(\bx_{1:N})}{\partial(\theta,\eta)^\top}
=
\begin{bmatrix}
\bJ_\theta & \bJ_\eta
\end{bmatrix}
\in\R^{L\times Np}
\end{equation}
denote the Jacobian in these local coordinates, with
$\bJ_\theta\in\R^{L\times r}$ and
$\bJ_\eta\in\R^{L\times (Np-r)}$. Under the Gaussian
model~\eqref{eq:obs_model}, the Fisher information matrix (FIM) for
$(\theta,\eta)$ is
\begin{equation}
\label{eq:FIM}
\bI(\theta,\eta)
=
\bJ^\top\bSigma^{-1}\bJ
=
\begin{pmatrix}
\bI_{\theta\theta} & \bI_{\theta\eta}\\
\bI_{\eta\theta} & \bI_{\eta\eta}
\end{pmatrix}.
\end{equation}
The effective Fisher information for $\theta$ after profiling out the
nuisance coordinates is the Schur complement
\begin{equation}
\label{eq:Ieff_def}
\bI_{\mathrm{eff}}^{(\theta)}
:=
\bI_{\theta\theta}
-
\bI_{\theta\eta}\,\bI_{\eta\eta}^{\dagger}\,\bI_{\eta\theta},
\end{equation}
where $(\cdot)^\dagger$ denotes the Moore--Penrose pseudoinverse. This is the local information quantity that organizes the rest of the
paper.

In whitened coordinates the same object has a simple geometric form.
Define the whitened Jacobian
\begin{equation}
\label{eq:whitened_J}
\widetilde{\bJ}
:=
\bSigma^{-1/2}\bJ
=
\begin{bmatrix}
\widetilde{\bJ}_\theta & \widetilde{\bJ}_\eta
\end{bmatrix},
\end{equation}
so that $\bI(\theta,\eta)=\widetilde{\bJ}^\top\widetilde{\bJ}$ is a
standard Gram matrix.

The nuisance-profiled Fisher information also admits the projector form
\begin{equation}
\label{eq:Ieff_projector}
\bI_{\mathrm{eff}}^{(\theta)}
=
\widetilde{\bJ}_\theta^\top
(\bI-\bP_\eta)
\widetilde{\bJ}_\theta,
\qquad
\bP_\eta:=\widetilde{\bJ}_\eta\widetilde{\bJ}_\eta^\dagger,
\end{equation}
where $\bP_\eta$ is the orthogonal projector onto
$\range(\widetilde{\bJ}_\eta)$. This follows by substituting
$\bI=\widetilde{\bJ}^\top\widetilde{\bJ}$ into~\eqref{eq:Ieff_def} and
using the Moore--Penrose identity
$\bA(\bA^\top\bA)^\dagger\bA^\top=\bA\bA^\dagger$; details are provided
in \appref{app:framework}{app:projector_variational}.

Equations~\eqref{eq:Ieff_def} and~\eqref{eq:Ieff_projector} give two
equivalent views of the same local information quantity. The
Schur-complement form is convenient for closed-form CRB calculations in
Sections~\ref{sec:lever1} and~\ref{sec:lever2}; the projector form makes
the identifiability condition explicit. Positive effective information
requires the footprint of $\theta$ in the whitened observation space to
have a nonzero component outside the nuisance span.

Whenever $\bI_{\mathrm{eff}}^{(\theta)}\succ \bm 0$, any locally
unbiased estimator $\hat\theta$ satisfies the vector Cram\'er--Rao bound
\begin{equation}
\label{eq:crb_vector}
\mathrm{Cov}(\hat\theta)\succeq
\bigl(\bI_{\mathrm{eff}}^{(\theta)}\bigr)^{-1}.
\end{equation}
For vector-valued parameters of interest we use the trace benchmark
\begin{equation}
\label{eq:crb_trace}
\mathrm{CRB}_\theta
:=
\operatorname{tr}\!\left(
\bigl(\bI_{\mathrm{eff}}^{(\theta)}\bigr)^{-1}
\right),
\end{equation}
in the scalar case $\theta=x_j$ we simply write
\begin{equation}
\label{eq:crb_scalar}
\mathrm{CRB}_j
=
\bigl(\bI_{\mathrm{eff}}^{(\theta)}\bigr)^{-1}.
\end{equation}
If $\bI_{\mathrm{eff}}^{(\theta)}$ is singular, then some directions of
$\theta$ are structurally non-identifiable from the release sequence
and the corresponding local benchmark diverges.

The labeled-target setting of Section~\ref{sec:lever1} uses this
framework with $r=1$ in the scalar case and with $r=p$ under the
coordinate-wise Kronecker lift for vector-valued records. For the unlabeled
multiset analysis of Section~\ref{sec:lever2}, we work with scalar records
so that the participating multiset admits a clean local sorted-coordinate
chart with $r=m$; \appref{app:framework}{app:block_jacobians} records the
vector-valued lift.

\subsubsection*{From the local CRB to operational risk}
Throughout the paper, the CRB serves as a local, nuisance-profiled benchmark for locally unbiased estimators. It is not the finite-noise risk of a constrained MLE/GLS estimator, nor is it a global Bayesian or minimax risk; nor is it the global DP guarantee itself. Section~\ref{sec:experiments} compares the local benchmark with finite-noise estimators in cases where the comparison is interpretable.

\subsection{Incremental Release Geometry}
\label{subsec:spine}

We next ask what one independent release can add to the profiled Fisher
information when the precision weights of the existing rows are held
fixed. The resulting identity exposes the identifiable-direction
criterion. Under a fixed total privacy budget, however, comparing
different values of $L$ also rescales all rows through $\kappa_L$. The
fixed-budget risks in Sections~\ref{sec:lever1}--\ref{sec:lever2}
combine row-level geometric gain with the composition-induced loss of
per-release precision.

Suppose the first $L$ releases induce the block Fisher matrix
\[
\bI^{(L)}
=
\begin{pmatrix}
\bI_{\theta\theta}^{(L)} & \bI_{\theta\eta}^{(L)}\\
\bI_{\eta\theta}^{(L)} & \bI_{\eta\eta}^{(L)}
\end{pmatrix},
\]
with
\[
\bI_{\mathrm{eff}}^{(\theta),L}
=
\bI_{\theta\theta}^{(L)}
-
\bI_{\theta\eta}^{(L)}
\bigl(\bI_{\eta\eta}^{(L)}\bigr)^\dagger
\bI_{\eta\theta}^{(L)}.
\]
Adding one more release appends a new whitened row
\[
\tilde{\bg}_{L+1}
=
\begin{bmatrix}
a_{L+1}^\top & b_{L+1}^\top
\end{bmatrix},
\]
where $a_{L+1}\in\R^r$ and $b_{L+1}\in\R^{Np-r}$ are the components
along the $\theta$- and $\eta$-blocks, respectively. When
$\bI_{\eta\eta}^{(L)}\succ\bm 0$, define the profiled increment
\begin{equation}
\label{eq:contrast_def}
d_{L+1}
:=
a_{L+1}
-
\bI_{\theta\eta}^{(L)}
\bigl(\bI_{\eta\eta}^{(L)}\bigr)^{-1}
b_{L+1}.
\end{equation}

The following result gives the single-row update obtained by appending
one additional scalar-output release to the stacked
model~\eqref{eq:obs_model}. Multi-output releases append multiple rows and
yield the corresponding higher-rank block updates.

\begin{theorem}
\label{thm:spine}
Under the Gaussian observation model~\eqref{eq:obs_model}, adding one
independent release gives the unconditional monotonicity
\[
\bI_{\mathrm{eff}}^{(\theta),L+1}
\succeq
\bI_{\mathrm{eff}}^{(\theta),L}.
\]
If $\bI_{\eta\eta}^{(L)}\succ\bm 0$, the exact rank-one update is
\begin{equation}
\label{eq:spine_update}
\bI_{\mathrm{eff}}^{(\theta),L+1}
=
\bI_{\mathrm{eff}}^{(\theta),L}
+
\frac{d_{L+1}d_{L+1}^\top}
{1+b_{L+1}^\top
\bigl(\bI_{\eta\eta}^{(L)}\bigr)^{-1}
b_{L+1}},
\end{equation}
and hence
\begin{equation}
\label{eq:spine_iff}
\bI_{\mathrm{eff}}^{(\theta),L+1}
=
\bI_{\mathrm{eff}}^{(\theta),L}
\quad\Longleftrightarrow\quad
d_{L+1}=\bm 0.
\end{equation}
\end{theorem}

Thus additional releases matter only through the component of the new
whitened release gradient that escapes the existing nuisance-profiled span.
Repeating the same profiled direction can improve conditioning along that
direction, but cannot enlarge the identifiable subspace.

\begin{proof}[Proof sketch]
Use the variational Schur-complement form of the profiled information.
Appending a whitened row adds a positive semidefinite quadratic term;
when the nuisance block is invertible, Sherman--Morrison gives
\eqref{eq:spine_update}. The full derivation is provided in
\appref{app:framework}{app:spine_rankone}.
\end{proof}

To see where such new directions can come from, write each release in
the form
\begin{equation}
\label{eq:deepsets_general}
q_\ell(\bx_{1:N})
=
\rho_\ell\!\left(\sum_{k\in C_\ell}\phi_\ell(\bx_k)\right),
\qquad
s_\ell:=\sum_{k\in C_\ell}\phi_\ell(\bx_k),
\end{equation}
with membership set $C_\ell$, per-record feature map
$\phi_\ell:\R^p\to\R^{d_\ell}$, and post-aggregation map
$\rho_\ell:\R^{d_\ell}\to\R$. The chain rule gives
\begin{equation}
\label{eq:deepsets_grad}
\frac{\partial q_\ell(\bx_{1:N})}{\partial \bx_k^\top}
=
\mathbb I(k\in C_\ell)\,
J_{\rho_\ell}(s_\ell)\,
J_{\phi_\ell}(\bx_k).
\end{equation}
Equation~\eqref{eq:deepsets_general} places the release family within the
standard class of permutation-invariant statistical queries represented by
Deep Sets--type decompositions~\cite{Zaheer2017DeepSets}. In this language,
the paper studies how repeated noisy statistical queries generate identifiable
directions after nuisance profiling under privacy-constrained precision
scaling. The chain-rule factorization in~\eqref{eq:deepsets_grad} shows that
each release contributes two pieces of geometry: an inclusion pattern and a
per-record gain profile. The technical results below analyze linear subset
averages, polynomial moments, and smooth threshold queries. Discrete histograms
and exact categorical counts fall outside the smooth Fisher setting unless
they are smoothed or analyzed with a different information bound.

Under~\eqref{eq:deepsets_general}--\eqref{eq:deepsets_grad}, changes in
the whitened release geometry can arise through two sources:
\begin{enumerate}[label=(\roman*),leftmargin=2em]
\item \emph{membership variation}: changes in the inclusion masks
$\{\mathbb I(k\in C_\ell)\}$ across releases;
\item \emph{feature diversity}: changes in $\phi_\ell$ and/or in the
local gain induced by $\rho_\ell$ across releases.
\end{enumerate}
Without either source of variation, repeated releases leave the identifiable
span unchanged; in the fully static case, this is seen directly from the
Jacobian factorization below.

In the fully static case, suppose $C_\ell\equiv C$, $\phi_\ell\equiv\phi$,
and $\rho_\ell\equiv\rho$ for all $\ell=1,\ldots,L$. Let
\[
s:=\sum_{k\in C}\phi(\bx_k),
\]
and let $\widetilde{\bJ}_k\in\R^{L\times p}$ denote the block of the
whitened Jacobian corresponding to record $\bx_k$. Then for every
participating record $k\in C$,
\begin{equation}
\label{eq:J_column_static}
\begin{aligned}
\widetilde{\bJ}_k
&=
\mathbb I(k\in C)\,
\widetilde{\bG}_{1:L}\,
J_\phi(\bx_k),\\
\widetilde{\bG}_{1:L}
&:=
\bSigma^{-1/2}\,
\mathrm{col}\!\bigl(J_\rho(s),\ldots,J_\rho(s)\bigr).
\end{aligned}
\end{equation}
All participating Jacobian blocks share the left factor
$\widetilde{\bG}_{1:L}$, which has rank at most one (exactly one when $J_\rho(s)\neq\bm 0$).
Repeated releases are then confined to a single observation-space
direction. They may sharpen estimation along that direction; the
identifiable subspace itself is unchanged.

\section{Labeled-Target Reconstruction under Subset Asymmetry}
\label{sec:lever1}

We first study reconstruction of a designated labeled record. Membership
variation can separate this target from a fixed nuisance background; the
question is whether repeated subset averages create new target-identifying
directions or only remeasure an existing profiled contrast. For scalar records
$(p=1)$, linear subset aggregates, and $\eta=x_{-j}$, we write
\[
\bI_{\mathrm{eff}}^{(j)}
:=
\bI_{\mathrm{eff}}^{(\theta)}
\]
evaluated at $\theta=x_j$.

\subsection{Linear Membership Design}
\label{subsec:lever1_linear}

For linear subset aggregates, stacking $L$ releases yields the
linear--Gaussian model
\begin{equation}
\label{eq:lever1_linear_model}
\by=\bA\bx+\bw,
\qquad
\bw\sim\mathcal N(\bm 0,\bSigma),
\quad
\bSigma\succ 0,
\end{equation}
where $\bx=(x_1,\ldots,x_N)^\top\in\R^N$, $\by\in\R^L$, and
$\bA\in\R^{L\times N}$ is the membership-induced design matrix.
We restrict attention to scalar records; the vector-valued extension
follows by the Kronecker lifting summarized in
\appref{app:framework}{app:block_jacobians}.
For subset averages,
\begin{equation}
\label{eq:lever1_subavg_H}
A_{\ell k}=\frac{1}{|C_\ell|}\,\mathbb I(k\in C_\ell).
\end{equation}
In the scalar subset-average case, the general Jacobian in Section~\ref{sec:framework} reduces to $\bJ=\bA$.

Let $\widetilde{\bA}:=\bSigma^{-1/2}\bA$ and partition
$\widetilde{\bA}=[\widetilde{\bA}_j\ \widetilde{\bA}_{-j}]$.
Then the identifiability condition becomes
\begin{equation}
\label{eq:lever1_ident_cond}
\widetilde{\bA}_j\notin\range(\widetilde{\bA}_{-j})
\quad\Longleftrightarrow\quad
\rank(\widetilde{\bA})>\rank(\widetilde{\bA}_{-j}).
\end{equation}
Geometrically, reconstruction of $x_j$ requires the target membership signature to leave a component outside the whitened nuisance span.

All subsequent quantities depend on $(\bA,\bSigma)$ only through the Gram matrix
\begin{equation}
\label{eq:lever1_gram}
\bG:=\bA^\top\bSigma^{-1}\bA=\widetilde{\bA}^\top\widetilde{\bA}.
\end{equation}
In this linear--Gaussian setting, the BLUE attains the nuisance-profiled
Cram\'er--Rao bound exactly; \appref{app:framework}{app:linear_blue_identity}
records the standard calculation. Consequently, the quantities denoted $\mathrm{MSE}^*$ in this
section are attainable oracle risks, not merely information-theoretic lower
bounds:
\begin{equation}
\label{eq:mse_optimal_s4}
\mathrm{MSE}^*(x_j)
=
\be_j^\top(\bA^\top\bSigma^{-1}\bA)^\dagger \be_j.
\end{equation}

Once initial identifiability is established, the remaining issue is whether
repeated membership-only releases create more than one profiled
target-background direction. We analyze this at the level of the restricted
Gram matrix.

\subsection{Nuisance-Symmetric One-Contrast Geometry}
\label{subsec:lever1_exchangeable}

The limitation of repeated membership-only releases can be stated before
choosing a particular release family. Let $C=\{j\}\cup B$ be a membership
block of size $m$, and reorder the coordinates so that the target index
comes first.

\appref{app:lever1}{app:one_contrast_gram} gives sufficient release-level
conditions for the symmetry class below. Here we work directly at the Gram level.

\begin{lemma}[One-contrast geometry under subset asymmetry]
\label{thm:lever1_onecontrast}
Suppose the restricted Gram matrix $\bG_C$ is invariant under all
permutations of the nuisance coordinates in $B$. Equivalently, it has
the block form
\begin{equation}
\label{eq:lever1_blocksym_gram}
\bG_C
=
\begin{bmatrix}
a & c\,\mathbf 1^\top\\
c\,\mathbf 1 & (d-e)\bI_{m-1}+e\,\mathbf 1\mathbf 1^\top
\end{bmatrix}
\end{equation}
for scalars $a,c,d,e$, corresponding respectively to target
self-information, target-background coupling, nuisance diagonal
information, and nuisance off-diagonal coupling. Let
\[
\bB:=(d-e)\bI_{m-1}+e\,\mathbf 1\mathbf 1^\top.
\]
If $\bB\succ 0$, then the profiled Fisher information for the target
coordinate is
\begin{equation}
\label{eq:lever1_onecontrast_ieff}
\bI_{\mathrm{eff}}^{(j)}
=
a-c^2\,\mathbf 1^\top \bB^{-1}\mathbf 1
=
a-c^2\frac{m-1}{d+(m-2)e}.
\end{equation}
\end{lemma}

\begin{proof}[Proof sketch]
The target-to-nuisance coupling is proportional to $\mathbf 1$, so
nuisance profiling depends only on the symmetric background average.
\appref{app:lever1}{app:one_contrast_gram} gives the block-inversion
calculation.
\end{proof}

\begin{remark}[Scope and release-level symmetry of the one-contrast claim]
\label{rem:lever1_onecontrast_scope}
Lemma~\ref{thm:lever1_onecontrast} is a Gram-level symmetry statement. A
sufficient release-level condition is exchangeability of the whitened
restricted Jacobian over the background set: for every permutation $\pi$ of
$B$, there is an orthogonal row transformation $U_\pi$ preserving the whitened
transcript such that $U_\pi\widetilde{\bJ}_j=\widetilde{\bJ}_j$ and
$U_\pi\widetilde{\bJ}_k=\widetilde{\bJ}_{\pi(k)}$ for $k\in B$, with a nonsingular
nuisance Gram block. Then all target--background inner products, background
diagonal entries, and background off-diagonal entries are constant, which
implies~\eqref{eq:lever1_blocksym_gram}. Heterogeneous record settings may
violate this exchangeability; the one-contrast conclusion is tied to this
idealized symmetry and should not be inferred for arbitrary membership
variation.
\end{remark}

The target couples to the nuisance coordinates through the one-dimensional
average direction $\mathrm{span}\{\mathbf 1\}$. Nuisance directions orthogonal
to $\mathbf 1$ do not enter the Schur complement. Repeated releases in this
class remeasure one profiled target-versus-background contrast; the family of
target-distinguishing directions does not grow with the transcript.

For illustration, in the fully exchangeable case $a=d=\alpha$ and
$c=e=\tau$, so that
\[
\bG_C=(\alpha-\tau)\bI_m+\tau\,\mathbf 1\mathbf 1^\top,
\qquad
\alpha>\tau\ge 0,
\]
the profiled information reduces to
\begin{equation}
\label{eq:ieff_exchangeable}
\bI_{\mathrm{eff}}^{(j)}
=
\frac{(\alpha-\tau)\bigl(\alpha+\tau(m-1)\bigr)}
{\alpha+\tau(m-2)}.
\end{equation}
This exchangeable formula is illustrative. The \textsf{IN/OUT} family below
shows the same one-contrast effect only after the background coordinates are
reduced to their aggregate sum. If the individual background coordinates are
kept as separate nuisance variables, the full $N$-dimensional nuisance block
is singular because all releases depend on the background only through
$\sum_{k\in B}x_k$. Lemma~\ref{thm:lever1_onecontrast} gives the Gram-level
geometry; the exact \textsf{IN/OUT} risk is computed independently in the
reduced identifiable coordinates
$(\theta_1,\theta_2)=(x_j,\sum_{k\in B}x_k)$.

\subsection{Repeated \textsf{IN/OUT} as a Canonical Instance}
\label{subsec:lever1_inout}

Fix a target index $j$ and a background set
$B\subseteq [N]\setminus\{j\}$ with $|B|=n$. Define the two
subset-average queries
\begin{equation}
\label{eq:q_in_fixed_background}
q_{\mathrm{in}}(\bx)=(n+1)^{-1}\Bigl(x_j+\sum_{k\in B}x_k\Bigr),
\end{equation}
\begin{equation}
\label{eq:q_out_fixed_background}
q_{\mathrm{out}}(\bx)=n^{-1}\sum_{k\in B}x_k.
\end{equation}
Across $L$ releases, suppose $L_{\mathrm{in}}$ copies of $q_{\mathrm{in}}$
and $L_{\mathrm{out}}$ copies of $q_{\mathrm{out}}$ are released, where
\[
L=L_{\mathrm{in}}+L_{\mathrm{out}},
\qquad
L_{\mathrm{in}},L_{\mathrm{out}}\ge 1,
\]
and assume independent Gaussian perturbations within each block with variances
$\sigma_{\mathrm{in}}^2$ and $\sigma_{\mathrm{out}}^2$.

Let $\bar y_{\mathrm{in}}$ and $\bar y_{\mathrm{out}}$ denote the empirical means
of the \textsf{IN} and \textsf{OUT} blocks:
\begin{equation}
\label{eq:block_means_fixed_background}
\bar y_{\mathrm{in}}
:=
L_{\mathrm{in}}^{-1}\sum_{\ell\in\mathcal I_{\mathrm{in}}} y_\ell,
\qquad
\bar y_{\mathrm{out}}
:=
L_{\mathrm{out}}^{-1}\sum_{\ell\in\mathcal I_{\mathrm{out}}} y_\ell.
\end{equation}
Since $q_{\mathrm{in}}$ and $q_{\mathrm{out}}$ depend on
$(x_j,\{x_k\}_{k\in B})$ only through the pair
$(x_j,\sum_{k\in B}x_k)$, the full $N$-dimensional model depends on
$\bx$ only through
\begin{equation}
\label{eq:theta_fixed_background}
\theta_1:=x_j,
\qquad
\theta_2:=\sum_{k\in B}x_k.
\end{equation}
Here $\theta_1$ is the target itself. Because
$L_{\mathrm{in}},L_{\mathrm{out}}\ge1$, the two mean rows are linearly
independent, so $x_j$ is identifiable in the reduced experiment. Without an
\textsf{OUT} block, $x_j$ is confounded with the background sum; without an
\textsf{IN} block, the target is absent. The resulting two-parameter
Gaussian experiment is
\begin{equation}
\label{eq:reduced_model_fixed_background}
\bar y_{\mathrm{in}}
=
(n+1)^{-1}(\theta_1+\theta_2)+\bar w_{\mathrm{in}},
\qquad
\bar y_{\mathrm{out}}
=
n^{-1}\theta_2+\bar w_{\mathrm{out}},
\end{equation}
with independent noises
\begin{equation}
\label{eq:reduced_noise_variances}
\bar w_{\mathrm{in}}\sim\mathcal N\!\bigl(0,\sigma_{\mathrm{in}}^2/L_{\mathrm{in}}\bigr),
\qquad
\bar w_{\mathrm{out}}\sim\mathcal N\!\bigl(0,\sigma_{\mathrm{out}}^2/L_{\mathrm{out}}\bigr).
\end{equation}

Here $\theta_1=x_j$, and the reduced design matrix has rank two iff both
\textsf{IN} and \textsf{OUT} blocks are present; with only \textsf{IN}
releases, $x_j$ is confounded with the background sum $\theta_2$.

Thus the repeated \textsf{IN/OUT} family reduces to a two-parameter
Gaussian experiment in $(\theta_1,\theta_2)$. The block means
$\bar y_{\mathrm{in}}$ and $\bar y_{\mathrm{out}}$ are sufficient for
estimating $\theta_1$ in the presence of the nuisance parameter
$\theta_2$; the sufficiency calculation is provided in
\appref{app:lever1}{app:lever1_fixed_background_limit}.

\begin{proposition}
\label{prop:fixed_background_crb}
In the reduced model~\eqref{eq:reduced_model_fixed_background}--\eqref{eq:reduced_noise_variances},
the minimum achievable MSE for unbiased estimation of $\theta_1=x_j$ is
\begin{equation}
\label{eq:crb_fixed_background}
\mathrm{MSE}^*(x_j)
=
(n+1)^2\sigma_{\mathrm{in}}^2/L_{\mathrm{in}}
+
n^2\sigma_{\mathrm{out}}^2/L_{\mathrm{out}},
\end{equation}
and is attained by the BLUE
\begin{equation}
\label{eq:blue_fixed_background}
\hat x_j=(n+1)\bar y_{\mathrm{in}}-n\bar y_{\mathrm{out}}.
\end{equation}
For the continuous relaxation
\[
L_{\mathrm{in}},L_{\mathrm{out}}>0,
\qquad
L_{\mathrm{in}}+L_{\mathrm{out}}=L,
\]
the optimal split satisfies
\begin{equation}
\label{eq:opt_split_s4}
L_{\mathrm{in}}/L_{\mathrm{out}}
=
\bigl((n+1)/n\bigr)\,\sigma_{\mathrm{in}}/\sigma_{\mathrm{out}},
\end{equation}
with minimized value
\begin{equation}
\label{eq:mse_inout_opt_s4}
\min_{\substack{L_{\mathrm{in}},L_{\mathrm{out}}>0\\
L_{\mathrm{in}}+L_{\mathrm{out}}=L}}\mathrm{MSE}^*(x_j)
=
L^{-1}\Bigl((n+1)\sigma_{\mathrm{in}}+n\sigma_{\mathrm{out}}\Bigr)^2.
\end{equation}
If release counts are required to be integers, the optimal integer split
is obtained by the nearest feasible rounding around this continuous
ratio.
\end{proposition}

\begin{proof}[Proof sketch]
The block means are sufficient for the two Gaussian blocks. In the
resulting two-parameter linear model, the Schur complement gives
\eqref{eq:crb_fixed_background}, and
\eqref{eq:blue_fixed_background} attains it. The continuous split
follows from Cauchy--Schwarz. Details are in
\appref{app:lever1}{app:fixed_background_crb_blue}.
\end{proof}

Equation~\eqref{eq:crb_fixed_background} matches the one-contrast picture
motivated by Lemma~\ref{thm:lever1_onecontrast}. Its algebra is computed in
the reduced experiment, since the full-background nuisance block is singular.
The \textsf{OUT} query improves estimation of $\theta_2$; it adds no second
target-distinguishing direction.

\subsection{Fixed-Budget Scaling under Uniform Splitting}
\label{subsec:lever1_uniform}

Under a fixed total privacy budget, recall from
\eqref{eq:kappaL_sigma_def} that uniform splitting gives
\[
\varepsilon_\ell=\varepsilon_{\mathrm{tot}}/L,
\qquad
\delta_\ell=\delta_{\mathrm{tot}}/L,
\qquad
\sigma_\ell^2=\kappa_L\,\Delta_{2,\ell}^2.
\]
For the repeated \textsf{IN/OUT} family, the two subset averages have
sensitivities
\[
\Delta_{2,\mathrm{in}}=\Delta/(n+1),
\qquad
\Delta_{2,\mathrm{out}}=\Delta/n,
\]
and hence
\begin{equation}
\label{eq:uniform_gaussian_calibration_fixed_background}
\sigma_{\mathrm{in}}^2=\kappa_L\,\Delta^2/(n+1)^2,
\qquad
\sigma_{\mathrm{out}}^2=\kappa_L\,\Delta^2/n^2.
\end{equation}
This calibration is for record-level adjacency of each released query, not
only for perturbations of the target coordinate. The \textsf{OUT} block
consumes privacy budget and enters the target risk through nuisance
estimation.

Substituting~\eqref{eq:uniform_gaussian_calibration_fixed_background}
into~\eqref{eq:crb_fixed_background} yields
\begin{equation}
\label{eq:dp_to_risk_fixed_total_closed_s4}
\mathrm{MSE}^*(x_j)
=
2\ln(1.25L/\delta_{\mathrm{tot}})
\bigl(L/\varepsilon_{\mathrm{tot}}\bigr)^2
\Delta^2
\bigl(L_{\mathrm{in}}^{-1}+L_{\mathrm{out}}^{-1}\bigr),
\end{equation}
where the dependence on $n$ vanishes.\footnote{The factors $(n+1)$ and
$n$ in the BLUE coefficients cancel the sensitivities $\Delta/(n+1)$
and $\Delta/n$ of the corresponding releases, so the fixed-budget target
risk depends on release counts and privacy allocation, not on block
size.}

\begin{proposition}
\label{prop:lever1_fixed_background_limit}
Consider the repeated \textsf{IN/OUT} family under uniform Basic Composition
with total budget $(\varepsilon_{\mathrm{tot}},\delta_{\mathrm{tot}})$. Then any
unbiased estimator of $x_j$ satisfies
\begin{equation}
\label{eq:fixed_background_lower_bound}
\mathrm{MSE}(\hat x_j)
\ge
8L\ln(1.25L/\delta_{\mathrm{tot}})\Delta^2/\varepsilon_{\mathrm{tot}}^2.
\end{equation}
If $L$ is even, equality is attained by the balanced split
\[
L_{\mathrm{in}}=L_{\mathrm{out}}=L/2
\]
together with~\eqref{eq:blue_fixed_background} (for odd $L$, the nearest
balanced integer split changes the bound by the factor
$L^2/(L^2-1)$).
\end{proposition}

\begin{proof}[Proof sketch]
Substitute~\eqref{eq:uniform_gaussian_calibration_fixed_background}
into~\eqref{eq:crb_fixed_background} and use
$1/L_{\mathrm{in}}+1/L_{\mathrm{out}}\ge 4/L$, with equality at the
balanced split. Attainability follows from
Proposition~\ref{prop:fixed_background_crb}.
\end{proof}

The risk in~\eqref{eq:fixed_background_lower_bound} grows as
$\Theta(L\log L)$. The $1/L$ averaging gain is dominated by the
$L^2\log L$ increase in per-release noise variance under Basic Composition.
For this fixed-background membership-only family, copying releases increases
the optimal unbiased reconstruction risk under Basic Composition.

\paragraph{No-new-geometry implication}
For the balanced fixed-background design, the solved risk can be written as
\begin{equation}
\label{eq:one_contrast_kappa_over_L}
\mathrm{MSE}^*(x_j)=4\Delta^2\kappa_L/L.
\end{equation}
The geometry supplies only the $1/L$ averaging factor; the accountant supplies
$\kappa_L$. Hence the sharp invariant is not that Basic Composition is the
only admissible accounting rule, but that copies are not a new contrast.
Table~\ref{tab:accountant_scaling_main} gives the two standard plug-ins.
For Gaussian zCDP, $\rho_\ell=\Delta^2/(2\sigma_\ell^2)$, so uniform splitting
of a fixed total $\rho_{\rm tot}$ gives $\kappa_L=\Theta(L)$; fixed-order RDP
has the same linear scaling up to constants.

\begin{table}[t]
\caption{Two accountant plug-ins for the same one-contrast geometry.}
\label{tab:accountant_scaling_main}
\centering
\scriptsize
\begin{tabular}{@{}p{0.28\columnwidth}@{\hspace{2pt}}p{0.22\columnwidth}@{\hspace{2pt}}p{0.21\columnwidth}@{\hspace{2pt}}p{0.19\columnwidth}@{}}
\hline
Accountant & $\kappa_L$ & risk $\kappa_L/L$ & effect of copies\\
\hline
Basic Composition & $\Theta(L^2\log L)$ & $\Theta(L\log L)$ & worse\\
linear zCDP/RDP & $\Theta(L)$ & $\Theta(1)$ & flat\\
\hline
\end{tabular}
\end{table}

Non-uniform allocation changes constants and yields a cube-root allocation
law. The one-contrast limitation remains.

\subsection{Non-Uniform Allocation}
\label{subsec:lever1_nonuniform_alloc}

Use the same repeated \textsf{IN/OUT} family, now with block-specific
per-release budgets $(\varepsilon_{\mathrm{in}},\delta_{\mathrm{in}})$ and
$(\varepsilon_{\mathrm{out}},\delta_{\mathrm{out}})$.
Under Basic Composition,
\begin{equation}
\label{eq:inout_basic_comp}
L_{\mathrm{in}}\varepsilon_{\mathrm{in}}
+
L_{\mathrm{out}}\varepsilon_{\mathrm{out}}
=
\varepsilon_{\mathrm{tot}},
\qquad
L_{\mathrm{in}}\delta_{\mathrm{in}}
+
L_{\mathrm{out}}\delta_{\mathrm{out}}
\le
\delta_{\mathrm{tot}}.
\end{equation}

\begin{proposition}
\label{prop:lever1_nonuniform}
For this repeated \textsf{IN/OUT} family, write
$c_{\mathrm{in}}:=\ln(1.25/\delta_{\mathrm{in}})$ and
$c_{\mathrm{out}}:=\ln(1.25/\delta_{\mathrm{out}})$. The minimum
achievable MSE satisfies
\begin{equation}
\label{eq:mse_inout_eps_delta}
\mathrm{MSE}^*(x_j)
=
2\Delta^2\left(
c_{\mathrm{in}}L_{\mathrm{in}}^{-1}\varepsilon_{\mathrm{in}}^{-2}
+
c_{\mathrm{out}}L_{\mathrm{out}}^{-1}\varepsilon_{\mathrm{out}}^{-2}
\right),
\end{equation}
which is independent of $n$.
Let $u:=L_{\mathrm{in}}\varepsilon_{\mathrm{in}}$ and
$v:=L_{\mathrm{out}}\varepsilon_{\mathrm{out}}$.
For fixed $(L_{\mathrm{in}},L_{\mathrm{out}})$ and fixed
$(\delta_{\mathrm{in}},\delta_{\mathrm{out}})$, the optimal
$\varepsilon$-allocation satisfies
\begin{equation}
\label{eq:opt_u_v_ratio}
u/v
=
\bigl(c_{\mathrm{in}}L_{\mathrm{in}}/
      (c_{\mathrm{out}}L_{\mathrm{out}})\bigr)^{1/3},
\end{equation}
Under symmetric calibration and joint comparison over release counts
with fixed total failure budget $\delta_{\mathrm{tot}}$, the global
minimum is attained at $L_{\mathrm{in}}=L_{\mathrm{out}}=1$.
\end{proposition}

\begin{proof}[Proof sketch]
With $u=L_{\mathrm{in}}\varepsilon_{\mathrm{in}}$ and
$v=L_{\mathrm{out}}\varepsilon_{\mathrm{out}}$, the KKT condition gives
$c_{\mathrm{in}}L_{\mathrm{in}}/u^3=c_{\mathrm{out}}L_{\mathrm{out}}/v^3$.
The fixed-$\delta_{\mathrm{tot}}$ release-count comparison is monotone
as described below. Full details are in
\appref{app:lever1}{app:lever1_nonuniform}.
\end{proof}

The independence of~\eqref{eq:mse_inout_eps_delta} from $n$ is the same
sensitivity-cancellation effect noted above. Under symmetric calibration,
the optimized risk takes the form

\begin{equation}
\label{eq:mse_inout_opt_eps_symmetric}
\min_{\varepsilon}\mathrm{MSE}^*(x_j)
=
2c\,\Delta^2/\varepsilon_{\mathrm{tot}}^{\,2}
\big(L_{\mathrm{in}}^{1/3}+L_{\mathrm{out}}^{1/3}\big)^3,
\end{equation}
where $c:=\ln(1.25/\delta)$.
For the joint comparison over release counts under a fixed total failure
budget $\delta_{\mathrm{tot}}$, use a feasible symmetric allocation
\[
\delta_{\mathrm{in}}=\delta_{\mathrm{out}}
=\delta_{\mathrm{tot}}/(L_{\mathrm{in}}+L_{\mathrm{out}}).
\]
Then $c$ becomes
\[
c(L_{\mathrm{in}},L_{\mathrm{out}})
=
\ln\!\bigl(1.25(L_{\mathrm{in}}+L_{\mathrm{out}})/\delta_{\mathrm{tot}}\bigr),
\]
which is nondecreasing in each release count. More explicitly, for every
integer $L_{\mathrm{in}},L_{\mathrm{out}}\ge 1$,
\[
c(L_{\mathrm{in}},L_{\mathrm{out}})\ge c(1,1),
\qquad
\big(L_{\mathrm{in}}^{1/3}+L_{\mathrm{out}}^{1/3}\big)^3\ge 8.
\]
Combining these two inequalities in
\eqref{eq:mse_inout_opt_eps_symmetric} gives
\[
\min_{\varepsilon}\mathrm{MSE}^*(x_j;L_{\mathrm{in}},L_{\mathrm{out}})
\ge
\min_{\varepsilon}\mathrm{MSE}^*(x_j;1,1),
\]
which proves that the global minimum is attained at
$L_{\mathrm{in}}=L_{\mathrm{out}}=1$. The cube-root law mainly confirms the
one-contrast conclusion; it is not a separate design principle.

\section{Unlabeled Multiset Reconstruction under Static Membership and Feature Diversity}
\label{sec:lever2}

We now fix the membership set. Subset asymmetry is unavailable, so a labeled
target cannot be identified from membership changes. The release sequence may
still vary through its per-record feature maps. Because the releases are
permutation invariant over the participating records, labels are unidentified;
away from collisions, the local target is the sorted participating multiset.
Feature diversity helps only if the weakest identifiable directions grow
faster than the accounting penalty.

\paragraph{Sorted quotient parameter and local benchmark}
Fix the membership set $C$ across all $L$ releases and suppose that each
release function $q_\ell$ is permutation-invariant over~$C$, i.e.,
$q_\ell(\bx_C)=q_\ell(\pi\bx_C)$ for every $\pi\in\mathfrak{S}_m$ and
$\ell=1,\dots,L$. Assume moreover that the participating values
$\{x_k\}_{k\in C}$ are pairwise distinct, and let
\[
z:=\bx_C^\uparrow=(z_1,\dots,z_m)
\in
\mathcal W:=\{u\in\R^m:u_1<\cdots<u_m\}
\]
denote the sorted representative. Whenever $\bI_C(L)$ denotes the Fisher
information matrix in this sorted chart, define the local unlabeled
benchmark by
\begin{equation}
\label{eq:crb_unlab_def}
\mathrm{CRB}_{\rm unlab}(L)
:=
\tr\!\bigl(\bI_C(L)^{-1}\bigr).
\end{equation}

\begin{proposition}
\label{prop:lever2_sorted_chart_main}
Under the assumptions above, the following hold on any neighborhood that
does not intersect a collision hyperplane:
\begin{enumerate}[label=(\roman*)]
\item the joint observation law is permutation-invariant,
\begin{equation}
\label{eq:perm_nonident}
p(\by\mid\bx_C)=p(\by\mid\pi\bx_C)
\qquad\text{for all }\pi\in\mathfrak{S}_m,
\end{equation}
so the labeled vector $\bx_C$ is not globally identifiable from
$(y_1,\dots,y_L)$;
\item the quotient by permutations admits a local chart given by the
sorted representative $z$;
\item the unlabeled matching loss
\[
\min_{\pi\in\mathfrak{S}_m}\|\hat{\bx}_C-\pi(\bx_C)\|_2^2
\]
coincides locally with the squared error in sorted coordinates
\[
\|\hat z-z\|_2^2.
\]
\end{enumerate}
\end{proposition}

The proof is provided in \appref{app:lever2}{app:sorted_chart}.

Away from collisions, the locally recoverable object is the sorted vector
$\theta=z=\bx_C^\uparrow\in\R^m$. In this static-membership quotient
problem, there are no nuisance coordinates inside the sorted chart:
coordinates outside $C$ do not enter the release law, so the Schur
complement of Section~\ref{subsec:fim_crb} reduces to
$\bI_{\theta\theta}=\bI_C(L)$ itself. Feature diversity breaks the
rank-one fully static factorization in
\eqref{eq:J_column_static}; if the features were also fixed, repeated
static releases would remain confined to one observation-space direction.
\subsection{A General Design Law for Static-Membership Releases}
\label{subsec:lever2_design_law}
For the remainder of this section, consider scalar-output releases
observed in sorted coordinates:
\begin{equation}
\label{eq:lever2_sorted_model}
y_\ell
=
q_\ell(z)+w_\ell,
\qquad
\ell=1,\dots,L,
\end{equation}
where $z=\bx_C^\uparrow\in\mathcal W$ and the noises are independent
Gaussians with variances
\[
w_\ell\sim\mathcal N(0,\sigma_\ell^2),
\qquad
\sigma_\ell^2=\kappa_L\,\Delta_{2,\ell}^2.
\]
The structural assumption is a common composition penalty $\kappa_L$ for
all $L$ releases. Uniform Basic Composition has this form, although the
factorization below does not require its specific growth law.

Define the sensitivity-normalized gradient
\begin{equation}
\label{eq:lever2_g_def}
g_\ell(z)
:=
\Delta_{2,\ell}^{-1}\nabla_z q_\ell(z)\in\R^m,
\end{equation}
the cumulative signal matrix
\begin{equation}
\label{eq:lever2_A_def}
\bA_L(z)
:=
\sum_{\ell=1}^L g_\ell(z)g_\ell(z)^\top,
\end{equation}
and note that the Fisher information matrix in sorted coordinates is\footnote{Unlike classical A-optimal experimental design, the noise scale $\kappa_L$ here depends on the design through privacy composition.}
\begin{equation}
\label{eq:lever2_fim_factorization}
\bI_C(L)
=
\frac{1}{\kappa_L}\,\bA_L(z).
\end{equation}
Whenever an identifiable stage exists, denote by
\[
L_{\mathrm{id}}
:=
\min\{L:\bA_L(z)\succ \bm 0\}
\]
the first stage at which the local unlabeled problem becomes full-rank.

\begin{theorem}
\label{thm:lever2_design_law}
Under the factorization \eqref{eq:lever2_fim_factorization},
\[
\rank\!\bigl(\bI_C(L)\bigr)=\rank\!\bigl(\bA_L(z)\bigr).
\]
Thus local unlabeled identifiability at stage $L$ is equivalent to
$\bA_L(z)\succ \bm 0$. Whenever $\bA_L(z)\succ \bm 0$,
\begin{equation}
\label{eq:lever2_crb_spectral}
\mathrm{CRB}_{\rm unlab}(L)
=
\kappa_L\,\tr\!\bigl(\bA_L(z)^{-1}\bigr).
\end{equation}
\end{theorem}

\begin{proof}[Proof sketch]
The factorization $\bI_C(L)=\kappa_L^{-1}\bA_L(z)$ preserves rank, and
inversion of this scalar multiple gives~\eqref{eq:lever2_crb_spectral}.
\end{proof}

\begin{corollary}
\label{cor:lever2_spectral_consequences}
Under the assumptions of Theorem~\ref{thm:lever2_design_law}, the
following hold:
\begin{enumerate}[label=(\roman*)]
\item Assume an identifiable stage exists, with first full-rank stage
$L_{\mathrm{id}}<\infty$. If $\lambda_{\min}(\bA_L(z))=o(\kappa_L)$,
then
\[
\mathrm{CRB}_{\rm unlab}(L)\to\infty
\qquad\text{as }L\to\infty.
\]
Combined with rank deficiency for $L<L_{\mathrm{id}}$, the local
reconstruction risk admits a finite minimizer
$L^\star\in[L_{\mathrm{id}},\infty)$.
\item
\[
\mathrm{CRB}_{\rm unlab}(L)\to 0
\qquad\Longleftrightarrow\qquad
\frac{\lambda_{\min}(\bA_L(z))}{\kappa_L}\to\infty.
\]
\end{enumerate}
\end{corollary}

\begin{proof}[Proof sketch]
Apply Theorem~\ref{thm:lever2_design_law} and bound the trace of
$\bA_L(z)^{-1}$ above and below by the reciprocal extreme eigenvalues.
Rank deficiency before $L_{\mathrm{id}}$ and divergence along any
subcritical eigen-direction give the finite-minimizer claim.
\end{proof}

\begin{remark}[Mixed-case interpretation]
\label{rem:lever2_mixed_regime}
Write the positive eigenvalues of $\bA_L(z)$ as
$\lambda_1,\dots,\lambda_m$. Then
\[
\mathrm{CRB}_{\rm unlab}(L)
=
\kappa_L\sum_{k=1}^m \lambda_k^{-1}.
\]
A single subcritical direction with
$\lambda_j(\bA_L(z))/\kappa_L\to 0$ is enough to force
$\mathrm{CRB}_{\rm unlab}(L)\to\infty$, regardless of how quickly the
supercritical directions improve.
\end{remark}

The relevant comparison is spectral. The local unlabeled benchmark vanishes
only when every eigendirection of $\bA_L(z)$ outgrows $\kappa_L$. The
asymptotic rates below use uniform Basic Composition, but the design law is
not tied to it: replacing Basic Composition by a linear Gaussian zCDP/RDP
accountant simply replaces $\kappa_L=\Theta(L^2\log L)$ by
$\kappa_L=\Theta(L)$ in the same spectral test. Saturating feature families
can still fail; well-conditioned threshold grids can plateau instead of
worsen.

For the explicit asymptotic statements, take a total budget
$(\varepsilon_{\mathrm{tot}},\delta_{\mathrm{tot}})$. Recalling
\eqref{eq:kappaL_sigma_def},
\begin{equation}
\kappa_L
=
\frac{2L^2\ln(1.25L/\delta_{\mathrm{tot}})}
{\varepsilon_{\mathrm{tot}}^2}
=
\Theta(L^2\log L).
\end{equation}

\subsection{Polynomial Moments: Rank Restoration and Saturation}
\label{subsec:lever2_poly}

We first apply the spectral law in Theorem~\ref{thm:lever2_design_law} to
polynomial moments. Consider releasing the first $L$ moments of the
participating multiset,
\begin{equation}
\label{eq:poly_query_new}
q_\ell(z)
=
\frac{1}{m}\sum_{k=1}^m z_k^\ell,
\qquad
\ell=1,\dots,L,
\end{equation}
where $|z_k|\le R$ for every $k$. The Jacobian has entries
\[
(\bJ_C)_{\ell k}
=
\frac{\ell}{m}\,z_k^{\ell-1},
\]
and under fixed-cardinality swap adjacency the worst-case sensitivity
satisfies
\[
\Delta_{2,\ell}^{\rm swap}\le \frac{2R^\ell}{m}.
\]
Writing
\[
r_k:=\frac{z_k}{R},
\qquad
\bm\nu_\ell:=(r_1^{\ell-1},\dots,r_m^{\ell-1})^\top\in\R^m,
\]
the normalized gradients become
\[
g_\ell(z)
=
\Delta_{2,\ell}^{-1}\nabla_z q_\ell(z)
=
\frac{\ell}{2R}\,\bm\nu_\ell,
\]
This gives
\begin{equation}
\label{eq:I_poly_rankone_new}
\begin{aligned}
\bA_L(z)
&=
\frac{1}{4R^2}
\sum_{\ell=1}^L \ell^2\,\bm\nu_\ell\bm\nu_\ell^\top,\\
\bI_C^{\mathrm{poly}}(L)
&=
\frac{1}{4R^2\kappa_L}
\sum_{\ell=1}^L
\ell^2\,\bm{\nu}_\ell\bm{\nu}_\ell^\top.
\end{aligned}
\end{equation}

For each sorted coordinate, the cumulative signal takes the form
\begin{equation}
\label{eq:Tj_poly_def}
\begin{aligned}
e_j^\top \bA_L(z)e_j
&=
\frac{1}{4R^2}\,T_j(L),\\
T_j(L)&:=\sum_{\ell=1}^L \ell^2\,r_j^{2(\ell-1)},
\qquad
r_j:=\frac{z_j}{R}.
\end{aligned}
\end{equation}
If $|r_j|<1$, then
$T_j(\infty)=(1+r_j^2)/(1-r_j^2)^3<\infty$; if $|r_j|=1$, then
$T_j(L)=\Theta(L^3)$. This coordinate-wise dichotomy is the source of the
spectral trace bottleneck.

For distinct nodes, the weighted Vandermonde structure already shows
that
\[
\rank(\bJ_C)=\min\{L,m\}.
\]
In particular, $L\ge m$ releases suffice for local identifiability in
sorted coordinates.

The scalar case makes the saturation effect explicit. For
$m=1$ and $r:=z_1/R\in[-1,1]$,
\begin{equation}
\label{eq:m1_crb}
\mathrm{CRB}_{m=1}(L)
=
\frac{4\kappa_L R^2}{T_1(L)}.
\end{equation}
If $|r|<1$, then
\begin{equation}
\label{eq:m1_floor}
\mathrm{CRB}_{m=1}(L)
\;\ge\;
\frac{4\kappa_L R^2 (1-r^2)^3}{1+r^2}.
\end{equation}
In particular, the scalar CRB diverges as $L\to\infty$, so a finite
budget-dependent minimizer $L^\star$ must exist in the strict interior
case under the fixed-budget calibration.

The same formula gives a stopping-scale rule: with $q=r^2\in(0,1)$, the
tail $T_1(\infty)-T_1(L)$ is $q^L$ times a quadratic polynomial in $L$,
so useful moment orders scale as $1/\log(1/q)$ up to logarithmic factors.
For multisets, the operative scale is the slowest saturating,
weakest-conditioned eigendirection; the scalar stopping-scale calculation and
the general-$m$ determinant bound are collected in
\appref{app:query_families}{app:query_families}.

The scalar formula also yields a coordinate-wise lower bound for general
multisets via the Schur complement; see \appref{app:query_families}{app:m1_additional}.

\begin{proposition}
\label{prop:poly_interior_matrix}
Assume the points $\{z_k\}_{k=1}^m$ are pairwise distinct and satisfy
$|z_k|<R$ for every $k$. Then the matrix sequence $\bA_L(z)$ increases
monotonically to a finite positive-definite limit
\[
\bA_\infty(z)
=
\frac{1}{4R^2}
\sum_{\ell=1}^\infty \ell^2\,\bm\nu_\ell\bm\nu_\ell^\top
\succ \bm 0.
\]
Consequently, for all $L\ge m$, there exist constants
$0<c_-\le c_+<\infty$, depending on $z$ and $R$ but not on $L$, such
that
\[
c_-\,\kappa_L^{-1}\bI_m
\preceq
\bI_C^{\mathrm{poly}}(L)
\preceq
c_+\,\kappa_L^{-1}\bI_m.
\]
Consequently,
\[
\mathrm{CRB}_{\rm unlab}^{\mathrm{poly}}(L)=\Theta(\kappa_L).
\]
\end{proposition}

\begin{proof}[Proof sketch]
The positive semidefinite signal matrices increase to a finite limit
because all nodes are strict interior. Distinct nodes give full
weighted-Vandermonde rank by stage $m$, so the limit is positive
definite and uniformly well conditioned for $L\ge m$. The spectral law
then gives the stated $\Theta(\kappa_L)$ CRB scaling.
\end{proof}

For $|r_j|<1$, the signal $T_j(L)$ saturates, so
$\lambda_{\min}(\bA_L(z))$ stays bounded and $\mathrm{CRB}_{\rm unlab}$
eventually diverges with $\kappa_L$; a finite minimizer $L^\star$ exists.
For a boundary coordinate with $|r_j|=1$, the diagonal signal
$T_j(L)=\Theta(L^3)$ grows faster than $\kappa_L=\Theta(L^2\log L)$, so
the local Fisher benchmark suggests continued coordinate-wise improvement
in fully boundary configurations. For a full multiset, however, the trace
benchmark is still controlled by the slowest eigendirections of
$\bA_L(z)$, as described in Remark~\ref{rem:lever2_mixed_regime}.

\paragraph{Design implication}
For polynomial moments, the trace benchmark is governed by the slowest
identifiable eigendirections (Remark~\ref{rem:lever2_mixed_regime}). A finite
optimum appears when at least one identifiable direction fails to outgrow
$\kappa_L$; this criterion is accountant-parametric. Continued improvement is
predicted in cases where every identifiable eigendirection outgrows the
accountant factor, as in the fully boundary $m=2$ example. Finite-noise
constrained estimators near boundaries require a separate interpretation from
this local unbiased benchmark. Appendices~\ref{app:poly}--\ref{app:poly_general_m}
give the corresponding polynomial-moment derivations.

\subsection{Closed-form illustration: \texorpdfstring{$m=2$}{m=2}}
\label{subsec:lever2_m2}

The $m=2$ case gives a closed-form illustration of
Proposition~\ref{prop:poly_interior_matrix} and of the mixed
interior/boundary behavior captured by
Remark~\ref{rem:lever2_mixed_regime}.
Let $z_1<z_2$ with $|z_k|\le R$, and define
\[
\mu:=\frac{z_1+z_2}{2},
\qquad
\gamma:=\frac{z_2-z_1}{2}>0.
\]
Write
\begin{equation}
\label{eq:F_def}
F(s,L):=\sum_{\ell=1}^L \ell^2 s^{\ell-1},
\qquad |s|\le 1.
\end{equation}

For $m=2$, writing $r_k:=z_k/R$ and $\rho:=z_1z_2/R^2$, the Fisher
information matrix takes the closed form
\begin{equation}
\label{eq:m2_fim}
\bI_C(L)
=
\frac{1}{4\kappa_L R^2}
\begin{pmatrix}
F(r_1^2,L) & F(\rho,L)\\
F(\rho,L) & F(r_2^2,L)
\end{pmatrix}.
\end{equation}
For $L\ge 2$ and $z_1\neq z_2$, the determinant
\[
D(L):=F(r_1^2,L)F(r_2^2,L)-F(\rho,L)^2
\]
is strictly positive, and inverting the $2\times 2$ matrix gives
\begin{equation}
\label{eq:m2_crb}
\mathrm{CRB}_{\rm unlab}(L)
=
\frac{4\kappa_L R^2\big(F(r_1^2,L)+F(r_2^2,L)\big)}{D(L)}.
\end{equation}
If $|r_1|,|r_2|<1$, then the benchmark is infinite at $L=1$ and
diverges as $L\to\infty$, so a finite minimizer exists over $L\ge2$.
Consecutive release counts can be compared by
\begin{equation}
\label{eq:m2_stop}
\frac{\kappa_{L+1}}{\kappa_L}
\;\ge\;
\frac{G(L)}{G(L+1)},
\qquad
G(L):=\frac{F(r_1^2,L)+F(r_2^2,L)}{D(L)},
\end{equation}
This inequality is exactly the condition that release count $L+1$ does
not improve over $L$. We use this comparison, or direct evaluation
of~\eqref{eq:m2_crb}, to locate $L^\star$ in the numerical experiments.

For $L=2$, substituting $F(s,2)=1+4s$ into \eqref{eq:m2_crb} yields
\begin{equation}
\label{eq:m2_L2}
\mathrm{CRB}_{\rm unlab}(2)
=
\frac{\kappa_2 R^2(R^2+4\mu^2+4\gamma^2)}{2\gamma^2},
\end{equation}
which scales as $\gamma^{-2}$ and makes the role of record separation
explicit. Two features are already visible at this first identifiable
stage: the risk is finite once two distinct moments are available, but
it blows up as the two records approach a collision $(\gamma\to 0)$,
which is the same local-chart instability behind the sorted
parameterization. At $L=2$, the noiseless moment equations also admit an algebraic inversion,
which we use as a finite-noise benchmark in Section~\ref{sec:experiments};
\appref{app:query_families}{app:m2} gives the details of the $m=2$ case.

\subsection{Threshold Queries: Bounded Sensitivity and Saturation}
\label{subsec:lever2_threshold}

The same spectral accounting appears for threshold queries. In this family,
budget-dependent saturation comes from localization of the threshold
derivative, in contrast with the geometric decay of high-order moments.
Consider smooth threshold queries
\begin{equation}
\label{eq:threshold_query}
q_\ell(z)
=
\frac{1}{m}\sum_{k=1}^m h(z_k-t_\ell),
\qquad
\ell=1,\dots,L,
\end{equation}
where $h:\R\to[0,1]$ is a smooth approximation of a step function and
$t_1<\cdots<t_L$ are threshold locations.
Define
\[
\bm\psi_\ell(z)
:=
\bigl(h'(z_1-t_\ell),\dots,h'(z_m-t_\ell)\bigr)^\top.
\]
Unlike the polynomial-moment family above, which was calibrated under
fixed-cardinality swap adjacency for moment releases, the threshold
family is calibrated using a uniform replacement/swap-style sensitivity
bound. Since $h\in[0,1]$, a pure value replacement within a fixed
participant set changes the averaged threshold statistic by at most
$1/m$. In the conservative swap/add-remove convention used for this
release family, at most two bounded summands can change, so we use the
uniform sensitivity bound $\Delta_{2,\ell}\le 2/m$. This choice affects
only constants and leaves the spectral accounting conclusions unchanged.
Hence
  \[
  g_\ell(z)=\frac{1}{2}\bm\psi_\ell(z),\qquad
  \bA_L(z)=\frac{1}{4}\sum_{\ell=1}^L \bm\psi_\ell(z)\bm\psi_\ell(z)^\top.
  \]
This expression also gives a useful conditioning heuristic. For interior
points away from the grid endpoints, a Riemann-sum approximation with
spacing $\eta_L$ gives
\[
\scalebox{0.94}{$
(\bA_L)_{ij}\approx(4\eta_L)^{-1}K(z_i-z_j),\quad
K(d):=\int h'(u)h'(u+d)\,du .
$}
\]
Thus, for $m=2$ with separation $d=|z_2-z_1|$ and negligible boundary
truncation, the weakest eigenvalue is approximately proportional to
$\eta_L^{-1}\{K(0)-K(d)\}$. More generally, if records are interior and the
kernel Gram matrix $[K(z_i-z_j)]_{i,j}$ is nonsingular, then
$\bA_L(z)=\Theta(L)$ and the lower bound in
Proposition~\ref{prop:threshold_truncation} is tight. Close records and
boundary truncation degrade the conditioning and explain the data-dependent
$L^\star$ seen in Fig.~\ref{fig:threshold}.
\begin{figure*}[t]
\centering
\resizebox{0.75\textwidth}{!}{\input{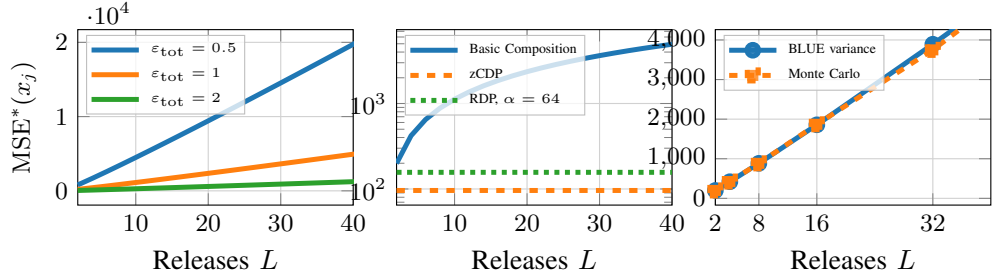}}
\caption{Labeled IN/OUT reconstruction under a fixed one-contrast geometry.
(Left) Closed-form $\mathrm{MSE}(x_j)$ under uniform Basic Composition for
three total privacy budgets. (Center) Same \textsf{IN/OUT} geometry with
three accountant plug-ins at the converted target
$(\varepsilon_{\rm tot},\delta_{\rm tot})=(1,10^{-5})$: Basic Composition,
zCDP, and fixed-order RDP. The geometry supplies the common $1/L$ averaging
factor in~\eqref{eq:one_contrast_kappa_over_L}; the accountant supplies
$\kappa_L$. (Right) Monte Carlo BLUE MSE versus the closed-form BLUE variance
under Basic Composition at $\varepsilon_{\rm tot}=1$.}
\label{fig:inout}
\label{fig:lever1_validation}
\end{figure*}
\begin{proposition}
\label{prop:threshold_truncation}
Suppose $h\in C^1$, $\operatorname{supp}(h')\subseteq [-b,b]$ for some $b>0$,
$\|h'\|_\infty\le M$, and the thresholds are placed on a regular grid
over $[-R,R]$ with spacing $\eta_L$ satisfying
\[
\frac{c_1}{L}\le \eta_L\le \frac{c_2}{L}
\]
for positive constants $c_1,c_2$ independent of $L$. Then:
\begin{enumerate}[label=(\roman*)]
\item The cumulative signal satisfies
\[
\tr\!\bigl(\bA_L(z)\bigr)=O(L).
\]

\item Assume an identifiable stage exists, with first full-rank stage
$L_{\mathrm{id}}<\infty$. Then for every $L\ge L_{\mathrm{id}}$,
\[
\mathrm{CRB}_{\rm unlab}(L)
=
\Omega\!\left(\frac{\kappa_L}{L}\right).
\]
Under uniform Basic Composition, this gives
\[
\mathrm{CRB}_{\rm unlab}(L)=\Omega(L\log L)\to\infty,
\]
and hence the local CRB benchmark has a finite budget-dependent optimum
under the fixed-budget Basic-Composition calibration. With linear zCDP/RDP
accounting, this generic lower bound is order-flat; any additional deterioration
must come from weak geometry, not from the accountant plug-in alone.
\end{enumerate}
\end{proposition}
\begin{proof}[Proof sketch]
For each coordinate, only thresholds within $b$ of $z_k$ contribute, and
the grid spacing $\eta_L=\Theta(1/L)$ gives $O(L)$ such thresholds.
Thus $\tr(\bA_L(z))=O(L)$; combining this with the spectral law gives
the $\Omega(\kappa_L/L)$ lower bound. Details are in
\appref{app:query_families}{app:threshold}.
\end{proof}
The optimal release count can depend strongly on where the records fall
relative to the threshold grid. Well-separated interior configurations usually
reach the stabilized kernel case earlier, whereas close or near-boundary
configurations can need more releases before the local inverse is well
conditioned. This data dependence is invisible to composition accounting alone
and is tested directly in Section~\ref{sec:experiments}.

The two sources of geometry can also be combined: feature-diverse releases can
anchor the participating multiset, while a small number of
membership-varying releases can resolve labels. A formal analysis of
optimal hybrid allocation is deferred to future work.

\section{Numerical Experiments}
\label{sec:experiments}

The numerical section is deliberately narrow. It checks the
geometry--accounting calculations for labeled \textsf{IN/OUT} releases and
for fixed-membership feature-diverse releases. We organize the experiments
around three predictions: one-contrast scaling for labeled \textsf{IN/OUT},
spectral saturation for feature-diverse static membership, and finite-noise
behavior near the local benchmark. The feature-diverse experiments use
polynomial moments, smooth thresholds, and normalized load profiles to test the
spectral conditioning predicted above. We report the loss associated with the
identifiable object in each case: single-coordinate MSE for $x_j$ in the linear--Gaussian \textsf{IN/OUT}
model, and sorted-coordinate error $\|\hat z-z\|_2^2$ for unlabeled multiset
reconstruction away from collisions. Unless otherwise noted, the numerical
panels instantiate the scalar accounting penalty $\kappa_L$ through the
Basic-Composition Gaussian calibration
in~\eqref{eq:gaussian_calibration}--\eqref{eq:kappaL_sigma_def}, with
$(\varepsilon_{\mathrm{tot}},\delta_{\mathrm{tot}})=(1,10^{-5})$.
This is a concrete calibration plug-in for the experiments, not an assumption
behind the Fisher-geometry conclusions: replacing it by a linear
zCDP/RDP-style Gaussian accountant changes the displayed rates through
$\kappa_L$, but leaves the profiled identifiable span unchanged.
Polynomial moments use fixed-cardinality swap sensitivity on $[-R,R]$ with
$R=1$, and threshold experiments use the bounded-sensitivity calibration of
Section~\ref{subsec:lever2_threshold}. To instantiate realistic operating
points, the real-data panels use electricity load-demand profiles from a
real-world grid-data corpus. Load profiles record power consumption over time
and are common inputs for power-system analysis and operational diagnostics;
fine-grained meter data can also expose consumer activity or
occupancy~\cite{Asghar2017SmartMeterPrivacy,Chen2013Occupancy,Ravi2022MeterDP}.
We use the profiles only as normalized scalar operating points for the local
Fisher experiments; they are clipped at the empirical 99th percentile and
normalized to $[-1,1]$ before pair and cohort sampling. Monte Carlo panels
average the independent noise trials stated in the captions.

\subsection{Labeled \textsf{IN/OUT} Reconstruction}
\label{subsec:experiments_inout}

For the labeled-target experiment, the setup follows
Section~\ref{subsec:lever1_inout}: $N=100$ records, background size
$n=N_{\mathrm{out}}=50$, record-level adjacency radius $\Delta=1$, and
$\delta_{\mathrm{tot}}=10^{-5}$. The scaling-law panels sweep
$L\in\{2,4,\dots,40\}$ under uniform Basic Composition for
$\varepsilon_{\mathrm{tot}}\in\{0.5,1,2\}$. The accountant-comparison and Monte
Carlo panels fix $\varepsilon_{\mathrm{tot}}=1$. For the zCDP curve, we choose
$\rho_{\rm tot}$ from the standard conversion to
$(\varepsilon_{\rm tot},\delta_{\rm tot})$-DP; for the fixed-order RDP curve,
we use order $\alpha=64$ and the corresponding converted RDP budget. The Monte
Carlo estimate uses $10^4$ Gaussian-noise trials at each tested $L$.

Figure~\ref{fig:inout} checks the solved one-contrast case. The left panel
matches the closed-form law \eqref{eq:dp_to_risk_fixed_total_closed_s4}: under
Basic Composition, the $\Theta(L^2\log L)$ noise penalty overwhelms the $1/L$
averaging gain. The center panel holds the \textsf{IN}/\textsf{OUT} geometry
fixed and changes only the accountant scalar in
\eqref{eq:one_contrast_kappa_over_L}. Basic Composition increases with $L$,
whereas the zCDP and fixed-order RDP plug-ins remain order-flat. Thus the
accountant changes the slope, not the target-identifying span: after the first
\textsf{IN}/\textsf{OUT} contrast, there is no new target direction. The right
panel shows that the empirical BLUE MSE tracks the attainable variance within
Monte Carlo error; in this linear--Gaussian model, this is an implementation
check, with no separate efficiency claim.

\FloatBarrier
\subsection{Static Membership with Feature Diversity}
\label{subsec:experiments_static}

We first test the feature-diversity prediction using polynomial moments.
Releasing the first $L$ moments makes $L$ both the number of releases and the
maximum moment order.
For the $m=2$ polynomial-moment benchmark~\eqref{eq:m2_crb}, we evaluate
$L\in\{2,\dots,40\}$ for three synthetic configurations: a strict-interior
pair $(-0.4,0.7)$, a mixed pair $(-0.6,1.0)$, and a fully boundary pair
$(-1.0,1.0)$.
\begin{figure}[t]
\centering
\resizebox{\columnwidth}{!}{
\definecolor{polyInterior}{HTML}{1E4D72}
\definecolor{polyMixed}{HTML}{B24A2F}
\definecolor{polyBoundary}{HTML}{2F6F4E}
\begin{tikzpicture}
\begin{groupplot}[
    group style={group name=figtwo, group size=2 by 1, horizontal sep=0.2cm},
    width=0.6\columnwidth,
    height=0.6\columnwidth,
    xmin=2, xmax=40,
    ymin=10, ymax=2e5,
    xlabel={Releases $L$},
    ylabel={$\mathrm{CRB}_{\mathrm{unlab}}(L)$},
    ylabel style={font=\small},
    ymode=log,
    grid=both,
    major grid style={draw=black!18},
    minor grid style={draw=black!8},
    tick align=outside,
    tick pos=left,
    legend style={
        draw=black!15,
        fill=white,
        fill opacity=0.78,
        text opacity=1,
        font=\small,
        cells={anchor=west},
        inner sep=0.7pt,
        row sep=0.3pt
    },
    label style={font=\small},
    xlabel style={font=\small, text width=0.30\columnwidth, align=center},
    tick label style={font=\small},
]

\nextgroupplot[
    legend style={at={(0.03,0.97)}, anchor=north west, legend columns=1},
]
\addplot+[
    smooth,
    no marks,
    color=polyInterior,
    line width=0.8pt
] coordinates {
  (2,378.012525656) (3,699.566542691) (4,1111.96723452)
  (5,1698.73414649) (6,2417.39896434) (7,3289.93001556)
  (8,4307.50635303) (9,5476.53961145) (10,6796.05952731)
  (11,8267.88475104) (12,9892.24303042) (13,11669.7818657)
  (14,13600.8688174) (15,15685.9501257) (16,17925.4500363)
  (17,20319.82116) (18,22869.5233806) (19,25575.0234035)
  (20,28436.7858074) (21,31455.2684676) (22,34630.9184918)
  (23,37964.1699786) (24,41455.442774) (25,45105.1420761)
  (26,48913.6585536) (27,52881.3687911) (28,57008.6359124)
  (29,61295.8102846) (30,65743.2302456) (31,70351.2228179)
  (32,75120.1043901) (33,80050.1813555) (34,85141.750707)
  (35,90395.100586) (36,95810.5107884) (37,101388.253231)
  (38,107128.592382) (39,113031.785656) (40,119098.08378)
};
\addlegendentry{interior}
\addplot+[
    only marks,
    forget plot,
    mark=star,
    mark size=2.3pt,
    color=polyInterior
] coordinates {(2,378.012525656)};

\addplot+[
    smooth,
    no marks,
    color=polyMixed,
    line width=0.8pt,
    dashed
] coordinates {
  (2,288.979276577) (3,345.406385249) (4,450.877752907)
  (5,613.929025798) (6,826.015711189) (7,1092.4228132)
  (8,1409.52174599) (9,1778.56667192) (10,2197.29590962)
  (11,2666.04802168) (12,3183.83778059) (13,3750.94419354)
  (14,4367.09194148) (15,5032.54515974) (16,5747.31541301)
  (17,6511.62694116) (18,7325.58703282) (19,8189.38247545)
  (20,9103.14582189) (21,10067.03712) (22,11081.1904909)
  (23,12145.7476673) (24,13260.8369447) (25,14426.5868431)
  (26,15643.1179648) (27,16910.5485802) (28,18228.9915792)
  (29,19598.5568055) (30,21019.3499934) (31,22491.4737956)
  (32,24015.0274785) (33,25590.1074201) (34,27216.807085)
  (35,28895.2173018) (36,30625.4263324) (37,32407.5200526)
  (38,34241.5820467) (39,36127.6937405) (40,38065.9344966)
};
\addlegendentry{mixed}
\addplot+[
    only marks,
    forget plot,
    mark=star,
    mark size=2.3pt,
    color=polyMixed
] coordinates {(2,288.979276577)};

\addplot+[
    smooth,
    no marks,
    color=polyBoundary,
    line width=0.8pt,
    densely dotted
] coordinates {
  (2,248.584323937) (3,161.716984442) (4,125.974688423)
  (5,104.85755444) (6,90.4431961603) (7,79.8115451311)
  (8,71.5775023194) (9,64.9798146684) (10,59.5579265242)
  (11,55.0137379252) (12,51.1443680298) (13,47.8062614208)
  (14,44.8946253572) (15,42.3309796392) (16,40.0552746977)
  (17,38.0207149316) (18,36.190252219) (19,34.5341473798)
  (20,33.028235291) (21,31.6526658094) (22,30.3909738198)
  (23,29.2293815528) (24,28.1562677741) (25,27.1617587927)
  (26,26.2374096938) (27,25.3759532799) (28,24.5711004324)
  (29,23.8173799557) (30,23.1100090447) (31,22.4447877266)
  (32,21.8180122351) (33,21.2264034538) (34,20.6670474476)
  (35,20.1373457542) (36,19.6349736127) (37,19.1578446819)
  (38,18.7040810998) (39,18.2719879606) (40,17.860031464)
};
\addlegendentry{boundary}

\nextgroupplot[
    ylabel={},
    yticklabel=\empty,
    legend style={at={(0.03,0.97)}, anchor=north west, legend columns=1},
]
\addplot+[
    smooth,
    no marks,
    color=polyInterior,
    line width=0.8pt
] coordinates {
  (2,371.5406620904) (3,838.4387842286) (4,1471.9092249599)
  (5,2331.2431214211) (6,3385.2496206161) (7,4657.8829917078)
  (8,6138.5119140905) (9,7834.7455031718) (10,9744.6934935344)
  (11,11871.0255321979) (12,14214.2987165783) (13,16775.9361813808)
  (14,19556.9091927188) (15,22558.2812983583) (16,25780.9880533004)
  (17,29225.9309490006) (18,32893.9486839286) (19,36785.8358558344)
  (20,40902.3429317626) (21,45244.1823823970) (22,49812.0317347170)
  (23,54606.5370390874) (24,59628.3155849715) (25,64877.9583788404)
  (26,70356.0322816435) (27,76063.0819081167) (28,81999.6313062877)
  (29,88166.1854551673) (30,94563.2316040536) (31,101191.2404756662)
  (32,108050.6673507436) (33,115141.9530492588) (34,122465.5248209634)
  (35,130021.7971561105) (36,137811.1725256158) (37,145834.0420586191)
  (38,154090.7861643117) (39,162581.7751039807) (40,171307.3695184472)
};
\addlegendentry{interior}
\addplot+[
    only marks,
    forget plot,
    mark=star,
    mark size=2.3pt,
    color=polyInterior
] coordinates {(2,371.5406620904)};
\node[anchor=west, text=polyInterior, font=\tiny] at (axis cs:2,371.5406620904) {$L^\star=2$};

\addplot+[
    smooth,
    no marks,
    color=polyMixed,
    line width=0.8pt,
    dashed
] coordinates {
  (2,2790.0678377428) (3,1980.4159241686) (4,821.0021857917)
  (5,1066.0779160369) (6,945.6165886577) (7,1203.0501946990)
  (8,1278.3951010175) (9,1573.2061820313) (10,1761.4429835437)
  (11,2106.0062960169) (12,2384.2785364944) (13,2784.7827639297)
  (14,3143.2180333710) (15,3603.0358130044) (16,4036.4211754383)
  (17,4557.5862206585) (18,5062.9800231853) (19,5646.8519953528)
  (20,6222.5715401900) (21,6870.1428559098) (22,7515.2340863079)
  (23,8227.2785335690) (24,8941.2173822414) (25,9718.3674579073)
  (26,10500.8879471794) (27,11343.6785024143) (28,12194.6724587422)
  (29,13103.5672810571) (30,14023.0248261297) (31,14998.4338901741)
  (32,15986.4073553206) (33,17028.6985897125) (34,18085.2800916720)
  (35,19194.7877219757) (36,20320.0948741318) (37,21497.1255230402)
  (38,22691.2921098165) (39,23936.1293716519) (40,25199.2991282683)
};
\addlegendentry{near-boundary}
\addplot+[
    only marks,
    forget plot,
    mark=star,
    mark size=2.3pt,
    color=polyMixed
] coordinates {(4,821.0021857917)};
\node[anchor=west, text=polyMixed, font=\tiny] at (axis cs:4,821.0021857917) {$L^\star=4$};
\end{groupplot}
\end{tikzpicture}}
\caption{Polynomial-moment CRB under uniform Basic Composition.
(Left) Synthetic $m=2$ pairs in strict-interior, mixed, and fully
boundary configurations.
(Right) Representative real-world load-profile cohorts. Stars
mark minimizing release counts $L^\star$ within the plotted
range.}
\label{fig:poly_synth}
\label{fig:lever2_dichotomy}
\end{figure}
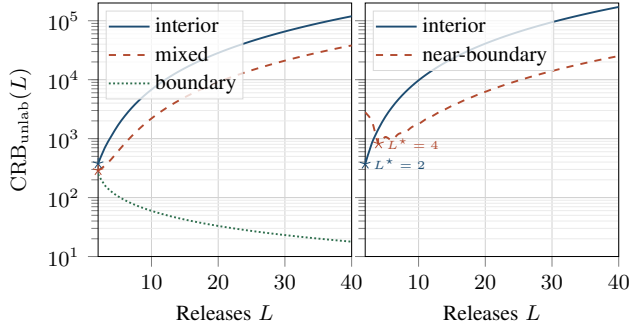
The left panel of Figure~\ref{fig:poly_synth} shows the predicted
interior/boundary split. The strict-interior and mixed configurations have
finite minimizing release counts, whereas the fully boundary configuration
continues decreasing over the plotted range. This is the trace-level effect predicted by the theory: the weakest
identifiable eigendirection of $\bA_L(z)$ controls the benchmark, while a
coordinate-wise interior/boundary classification is insufficient.

We next examine where normalized records from the real-world load profiles fall on the
interior/boundary spectrum for polynomial moments. The real data are used as
a source of operating points, not as a model of a deployed release
procedure. After the same clipping and normalization protocol used in the
synthetic study, the load profiles give realistic pair separations,
near-boundary behavior, and moment-map conditioning.

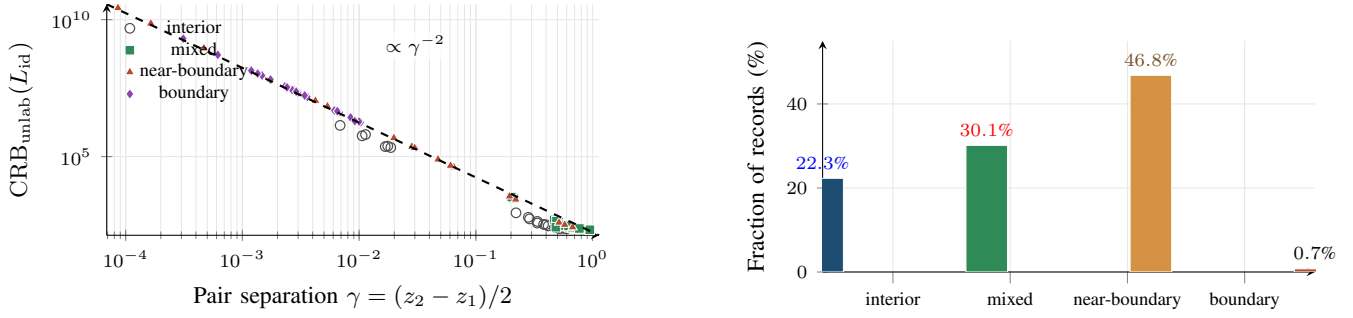
\begin{figure*}[t]
\centering
\begin{minipage}[t]{0.465\textwidth}
\centering
\definecolor{polyInterior}{HTML}{1E4D72}
\definecolor{polyMixed}{HTML}{2E8B57}
\definecolor{polyNearBoundary}{HTML}{B24A2F}
\definecolor{polyBoundary}{HTML}{8E44AD}
\begin{tikzpicture}
\begin{axis}[
  width=0.96\linewidth, height=0.55\linewidth,
  xmode=log, ymode=log,
  xlabel={Pair separation $\gamma = (z_2 - z_1)/2$},
  ylabel={$\mathrm{CRB}_{\mathrm{unlab}}(L_{\mathrm{id}})$},
  xlabel near ticks,
  ylabel near ticks,
  legend style={font=\scriptsize, draw=none, fill=none, at={(0.02,0.98)}, anchor=north west, row sep=-2pt},
  tick label style={font=\scriptsize},
  label style={font=\small},
  grid=both, grid style={black!10, line width=0.3pt},
  axis lines=left,
  clip mode=individual,
]
\addplot+[only marks, mark=o, mark size=1.8pt, mark options={draw=black!70, fill=white, line width=0.45pt}] coordinates {
  (0.0176043,241644) (0.33655,429.878) (0.0113453,639528) (0.558651,238.98) (0.380552,335.206) (0.397425,325.542) (0.598408,233.584) (0.519873,245.94) (0.0186773,213972) (0.516166,245.692) (0.0167458,231893) (0.505425,247.523) (0.503877,248.025) (0.52387,244.347) (0.0106139,573052) (0.00687537,1.37865e+06) (0.335794,373.842) (0.288655,530.273) (0.281672,607.85) (0.524838,247.901) (0.220765,885.655) (0.470249,265.036) (0.416134,294.909) (0.526645,243.84)
};
\addlegendentry{interior}
\addplot+[only marks, mark=square*, mark size=1.8pt, mark options={draw=white, fill=polyMixed, line width=0.3pt}] coordinates {
  (0.209277,3049.18) (0.203095,3323.58) (0.947482,212.785) (0.553692,318.909) (0.725987,242.03) (0.74296,238.217) (0.737074,228.439) (0.50281,385.264) (0.513853,398.701) (0.552915,333.752) (0.648863,264.566) (0.590841,287.881) (0.716486,243.342) (0.748126,236.455) (0.560054,300.431) (0.710229,248.039) (0.62974,266.349) (0.488281,439.946) (0.46842,485.764) (0.489996,437.295) (0.669629,263.795) (0.484697,263.901) (0.620617,294.33) (0.770782,237.376)
};
\addlegendentry{mixed}
\addplot+[only marks, mark=triangle*, mark size=2.0pt, mark options={draw=white, fill=polyNearBoundary, line width=0.3pt}] coordinates {
  (0.0283246,231665) (0.515825,391.659) (0.674282,264.899) (8.60456e-05,2.67458e+10) (0.577359,326.244) (0.00535636,6.88104e+06) (0.194405,3543.44) (0.0198517,477957) (0.000164434,7.29117e+09) (0.00149252,8.82852e+07) (0.0310202,188530) (0.219491,2726.3) (0.0618541,45345) (0.00174002,6.33355e+07) (0.0602363,47451.6) (0.0025385,3.02598e+07) (0.000468199,8.9858e+08) (0.0308045,195413) (0.00312159,2.04513e+07) (0.00422509,1.10198e+07) (0.064838,41118.1) (0.0298993,208716) (0.0608049,46773.7) (0.0474542,79119.2)
};
\addlegendentry{near-boundary}
\addplot+[only marks, mark=diamond*, mark size=2.0pt, mark options={draw=white, fill=polyBoundary, line width=0.3pt}] coordinates {
  (0.00612494,4.81923e+06) (0.00279341,2.54924e+07) (0.0103504,1.79297e+06) (0.00174212,6.50151e+07) (0.00683711,4.19186e+06) (0.00147372,9.15032e+07) (0.00658443,4.49694e+06) (0.00135304,1.07226e+08) (0.00031168,2.02997e+09) (0.00290743,2.30667e+07) (0.00267238,2.69886e+07) (0.00117068,1.42772e+08) (0.00650551,4.54985e+06) (0.00844426,2.69219e+06) (0.00298458,2.18613e+07) (0.00998687,1.9077e+06) (0.00360178,1.51608e+07) (0.000615861,5.19629e+08) (0.00231188,3.69265e+07) (0.00119237,1.37546e+08) (0.0028913,2.35739e+07) (0.00341762,1.69825e+07) (0.00918178,1.97024e+06) (0.0024105,3.37832e+07)
};
\addlegendentry{boundary}
\addplot+[no marks, dashed, color=black, line width=0.8pt, forget plot] coordinates {
  (6.88365e-05,3.66103e+10) (7.59319e-05,3.00879e+10) (8.37586e-05,2.47276e+10) (9.23921e-05,2.03222e+10) (0.000101915,1.67017e+10) (0.000112421,1.37262e+10) (0.000124008,1.12808e+10) (0.000136791,9.27102e+09) (0.00015089,7.61933e+09) (0.000166444,6.2619e+09) (0.0001836,5.1463e+09) (0.000202525,4.22945e+09) (0.0002234,3.47595e+09) (0.000246427,2.85669e+09) (0.000271828,2.34775e+09) (0.000299847,1.92948e+09) (0.000330754,1.58573e+09) (0.000364847,1.30323e+09) (0.000402453,1.07105e+09) (0.000443937,8.80234e+08) (0.000489696,7.23415e+08) (0.000540172,5.94534e+08) (0.00059585,4.88614e+08) (0.000657268,4.01564e+08) (0.000725017,3.30023e+08) (0.000799748,2.71227e+08) (0.000882183,2.22906e+08) (0.000973115,1.83194e+08) (0.00107342,1.50557e+08) (0.00118406,1.23734e+08) (0.00130611,1.0169e+08) (0.00144074,8.35735e+07) (0.00158925,6.86844e+07) (0.00175306,5.64478e+07) (0.00193376,4.63913e+07) (0.00213308,3.81264e+07) (0.00235295,3.13339e+07) (0.00259548,2.57516e+07) (0.00286301,2.11638e+07) (0.00315812,1.73933e+07) (0.00348365,1.42946e+07) (0.00384273,1.17479e+07) (0.00423882,9.65495e+06) (0.00467574,7.93486e+06) (0.0051577,6.52122e+06) (0.00568933,5.35942e+06) (0.00627576,4.40461e+06) (0.00692264,3.6199e+06) (0.0076362,2.97499e+06) (0.00842331,2.44498e+06) (0.00929155,2.00939e+06) (0.0102493,1.6514e+06) (0.0113057,1.3572e+06) (0.0124711,1.1154e+06) (0.0137566,916687) (0.0151745,753373) (0.0167386,619155) (0.018464,508849) (0.0203672,418194) (0.0224666,343690) (0.0247823,282459) (0.0273368,232138) (0.0301545,190781) (0.0332627,156792) (0.0366913,128859) (0.0404733,105902) (0.0446451,87034.5) (0.049247,71528.8) (0.0543232,58785.5) (0.0599226,48312.5) (0.0660991,39705.3) (0.0729124,32631.5) (0.0804279,26818) (0.0887181,22040.2) (0.0978627,18113.6) (0.10795,14886.6) (0.119077,12234.4) (0.131351,10054.8) (0.14489,8263.46) (0.159825,6791.28) (0.176299,5581.37) (0.194471,4587.01) (0.214516,3769.81) (0.236628,3098.19) (0.261018,2546.23) (0.287923,2092.6) (0.317601,1719.79) (0.350338,1413.4) (0.38645,1161.59) (0.426283,954.649) (0.470223,784.572) (0.518691,644.796) (0.572156,529.921) (0.631131,435.512) (0.696186,357.923) (0.767946,294.157) (0.847102,241.751) (0.934418,198.681) (1.03073,163.285) (1.13698,134.195)
};
\node[
  anchor=west,
  font=\scriptsize,
  black,
  fill=white,
  fill opacity=0.78,
  text opacity=1,
  inner sep=1pt
]
  at (rel axis cs:0.56,0.82) {$\propto \gamma^{-2}$};
\end{axis}
\end{tikzpicture}
\end{minipage}
\hfill
\begin{minipage}[t]{0.465\textwidth}
\centering
\definecolor{regimeInterior}{HTML}{1E4D72}
\definecolor{regimeMixed}{HTML}{2E8B57}
\definecolor{regimeNearBdry}{HTML}{D69243}
\definecolor{regimeBoundary}{HTML}{B24A2F}
\begin{tikzpicture}
\begin{axis}[
  width=0.96\linewidth,
  height=0.55\linewidth,
  ylabel={Fraction of records (\%)},
  ylabel near ticks,
  xmin=-0.6,
  xmax=3.6,
  ymin=0,
  ymax=55,
  xtick={0,1,2,3},
  xticklabels={
    {interior},
    {mixed},
    {near-boundary},
    {boundary}
  },
  xticklabel style={align=center, font=\scriptsize},
  tick label style={font=\scriptsize},
  label style={font=\small},
  axis lines=left,
  grid=both,
  grid style={black!10, line width=0.3pt},
  ybar,
  bar width=0.55cm,
  nodes near coords={\pgfmathprintnumber[fixed,fixed zerofill,precision=1]{\pgfplotspointmeta}\%},
  every node near coord/.append style={font=\scriptsize},
  nodes near coords align={vertical},
]
\addplot+[ybar, fill=regimeInterior, draw=white, line width=0.3pt] coordinates {(0,22.304)};
\addplot+[ybar, fill=regimeMixed, draw=white, line width=0.3pt] coordinates {(1,30.147)};
\addplot+[ybar, fill=regimeNearBdry, draw=white, line width=0.3pt] coordinates {(2,46.814)};
\addplot+[ybar, fill=regimeBoundary, draw=white, line width=0.3pt] coordinates {(3,0.735)};
\end{axis}
\end{tikzpicture}
\end{minipage}
\caption{Empirical operating bins from real-world electricity load profiles after normalization
to $[-1,1]$.
(Left) First-identifiable polynomial-moment CRB versus pair
separation $\gamma=(z_2-z_1)/2$; the dashed line is the
$\gamma^{-2}$ reference, and colors follow the magnitude bins
in the right panel.
(Right) Distribution of normalized record magnitudes across these
bins.}
\label{fig:realdata}
\label{fig:realdata_summary}
\end{figure*}

\begin{figure}[t]
\centering
\resizebox{\columnwidth}{!}{\input{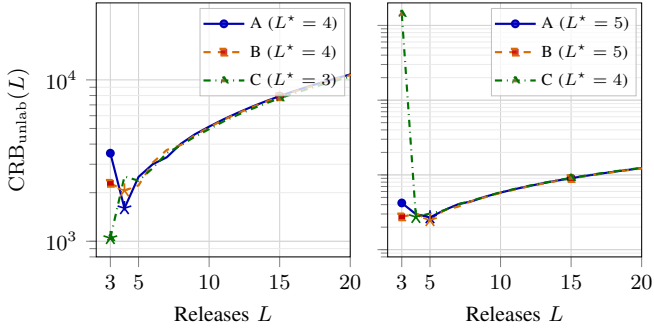}}
\caption{Smooth-threshold CRB for representative $m=3$ triples under
uniform Basic Composition. Top: interior triples. Bottom: near-boundary
triples. Stars mark minimizing release counts
$L^\star$ within the plotted range.}
\label{fig:threshold}
\label{fig:threshold_queries}
\end{figure}

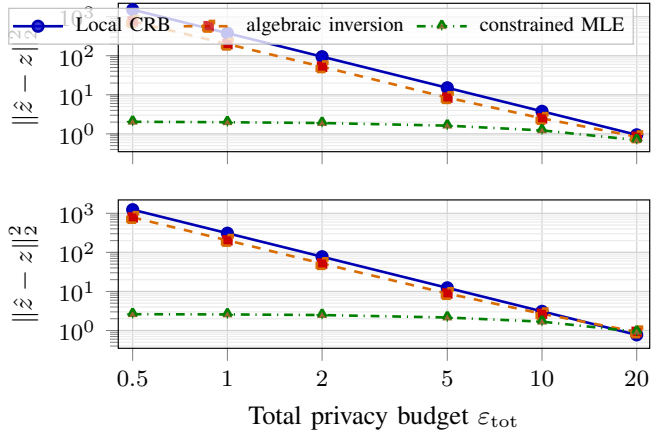
\begin{figure}[t]
\centering
\resizebox{\columnwidth}{!}{
\begin{tikzpicture}
\begin{groupplot}[
    group style={group name=expfive, group size=1 by 2, vertical sep=0.58cm},
    width=0.90\columnwidth,
    height=0.38\columnwidth,
    xmode=log,
    xmin=0.45, xmax=22,
    ymin=0.35, ymax=2200,
    ymode=log,
    grid=both,
    major grid style={draw=black!18},
    minor grid style={draw=black!8},
    tick align=outside,
    tick pos=left,
    xtick={0.5,1,2,5,10,20},
    xticklabels={$0.5$,$1$,$2$,$5$,$10$,$20$},
    label style={font=\small},
    tick label style={font=\footnotesize},
    legend style={
        draw=black!25,
        fill=white,
        fill opacity=0.78,
        text opacity=1,
        font=\scriptsize,
        cells={anchor=west},
        inner sep=1pt,
        row sep=1pt,
        column sep=2pt
    },
]

\nextgroupplot[
    ylabel={$\|\hat z - z\|_2^2$},
    xticklabels={},
    legend style={at={(0.97,0.97)}, anchor=north east, legend columns=3},
]
\addplot+[
    color=blue!70!black,
    mark=*,
    mark size=1.8pt,
    line width=0.9pt,
] coordinates {
    (0.5,1512.05010262) (1,378.012525656) (2,94.503131414) (5,15.1205010262) (10,3.78012525656) (20,0.94503131414)
};
\addlegendentry{Local CRB}
\addplot+[
    color=orange!85!black,
    dashed,
    mark=square*,
    mark size=1.8pt,
    line width=0.9pt,
    error bars/.cd, y dir=both, y explicit,
] coordinates {
    (0.5,758.098807191) +- (0,29.4232560201) (1,198.672405483) +- (0,7.61327994084) (2,52.3750213105) +- (0,2.04245381006) (5,8.49175377876) +- (0,0.296832882915) (10,2.495674019) +- (0,0.0793732216582) (20,0.841862227566) +- (0,0.0220364404881)
};
\addlegendentry{algebraic inversion}
\addplot+[
    color=green!50!black,
    dash dot,
    mark=triangle*,
    mark size=2.0pt,
    line width=0.9pt,
    error bars/.cd, y dir=both, y explicit,
] coordinates {
    (0.5,2.03999690294) +- (0,0.0271198380377) (1,1.98570454185) +- (0,0.0271135592351) (2,1.89587227926) +- (0,0.0273744816724) (5,1.62433406751) +- (0,0.0270958418823) (10,1.22700316613) +- (0,0.0247829049459) (20,0.694509561189) +- (0,0.0163443026913)
};
\addlegendentry{constrained MLE}

\nextgroupplot[
    ylabel={$\|\hat z - z\|_2^2$},
    xlabel={Total privacy budget $\varepsilon_{\mathrm{tot}}$},
    legend style={at={(0.97,0.97)}, anchor=north east, legend columns=3},
]
\addplot+[
    color=blue!70!black,
    mark=*,
    mark size=1.8pt,
    line width=0.9pt,
] coordinates {
    (0.5,1237.3975236) (1,309.349380899) (2,77.3373452248) (5,12.373975236) (10,3.09349380899) (20,0.773373452248)
};
\addplot+[
    color=orange!85!black,
    dashed,
    mark=square*,
    mark size=1.8pt,
    line width=0.9pt,
    error bars/.cd, y dir=both, y explicit,
] coordinates {
    (0.5,786.111180395) +- (0,30.405790936) (1,203.26639483) +- (0,8.13475990944) (2,51.8234374818) +- (0,1.91716322233) (5,8.86980681204) +- (0,0.308057555463) (10,2.67745549761) +- (0,0.0797258819451) (20,0.920847592235) +- (0,0.027095507463)
};
\addplot+[
    color=green!50!black,
    dash dot,
    mark=triangle*,
    mark size=2.0pt,
    line width=0.9pt,
    error bars/.cd, y dir=both, y explicit,
] coordinates {
    (0.5,2.61430157626) +- (0,0.0327276816567) (1,2.57852451799) +- (0,0.0331111438652) (2,2.48708843804) +- (0,0.0336331268468) (5,2.15192435002) +- (0,0.0345525223308) (10,1.68275615055) +- (0,0.0333859429097) (20,0.931871048787) +- (0,0.0247573997256)
};

\end{groupplot}
\end{tikzpicture}}
\caption{Finite-noise reconstruction for the $m=2$ polynomial-moment
release at $L=2$. The plots compare the local CRB, algebraic
moment inversion, and constrained Gaussian MLE as functions of
$\varepsilon_{\rm tot}$. Top: strict-interior pair. Bottom: boundary pair.}
\label{fig:mle}
\label{fig:lever2_monte_carlo_main}
\end{figure}
In Fig.~\ref{fig:realdata}, the left panel applies the $\gamma^{-2}$ scaling
in \eqref{eq:m2_L2} to real-data pairs. It plots the first-identifiable local
CRB at $L=L_{\mathrm{id}}=2$ against the pair separation
$\gamma=(z_2-z_1)/2$, with the dashed line giving the closed-form
$\gamma^{-2}$ reference. Each pair is assigned to the magnitude bin
containing $\max(|z_1|,|z_2|)$. The largest risks occur near collisions, as
expected from the $\gamma^{-2}$ factor in~\eqref{eq:m2_L2}. The right panel
summarizes the empirical distribution of normalized record magnitudes across
these descriptive bins; the thresholds are not theorem assumptions. Fully
boundary records are rare. Mixed and near-boundary records dominate the pool,
so finite data-dependent optima are the relevant case for these cohorts. The
near-boundary concentration should be interpreted relative to the
99th-percentile clipping and normalization to $[-1,1]$.

Representative real-data CRB-versus-$L$ curves for individual cohorts are
shown in the right panel of Fig.~\ref{fig:poly_synth}; they match the
synthetic mixed-case behavior in the left panel.

Smooth threshold releases give another feature-diverse fixed-membership
family with bounded sensitivity. We use logistic-smoothed releases
\[
q_\ell(\bx)
=
\frac{1}{m}\sum_{k\in C} h(x_k-t_\ell),
\]
where $h$ is a logistic smoothing of the step function with bandwidth
$\tau=0.12$ and the thresholds $t_\ell$ are placed on the midpoint grid
over $[-R,R]$. The exact local unlabeled CRB is evaluated for $m=3$,
$R=1$, and
$(\varepsilon_{\mathrm{tot}},\delta_{\mathrm{tot}})=(1,10^{-5})$.
Although Proposition~\ref{prop:threshold_truncation} is stated for
compactly supported $h'$, the logistic smoothing used here has localized
derivatives; a compact-support control is reported in
\appref{app:query_families}{app:logistic_smoothing}.

Figure~\ref{fig:threshold} shows finite minimizing release counts for smooth
threshold families. Additional thresholds initially improve local
identifiability; after the relevant eigendirections saturate, the composition
penalty dominates. Near-boundary triples have larger minimizing release counts
in the displayed cases, reflecting the data dependence of the threshold-grid
conditioning. These plots show the same finite-optimum behavior outside the
polynomial-moment family.

\subsection{Finite-Noise Reconstruction}
\label{subsec:experiments_mle}

The CRB used throughout is a local unbiased benchmark. In finite-noise cases,
especially with box constraints or boundary points, constrained MLEs can be
biased and may have MSE below the local unbiased CRB. We use MLE simulations as
an operational check of when the local benchmark predicts finite-noise
reconstruction.

We evaluate a constrained Gaussian MLE at the first identifiable stage
$L=2$ for the $m=2$ polynomial-moment release, using the strict-interior
pair $z=(-0.4,0.7)$ and the boundary pair $z=(-1.0,0.5)$. We sweep
$\varepsilon_{\mathrm{tot}}\in\{0.5,1,2,5,10,20\}$ with
$(\delta_{\mathrm{tot}},R)=(10^{-5},1)$, solve the box-constrained
Gaussian likelihood over $-1\le u_1<u_2\le1$, and report the
sorted-coordinate squared error $\|\widehat z-z\|_2^2$. Each point
averages $5{,}000$ Gaussian-noise trials.

Figure~\ref{fig:mle} shows that the constrained MLE follows the CRB trend at
moderate to large budgets, while small budgets expose bias and box-constraint
effects. The fixed-budget $L$-sweep in Fig.~\ref{fig:app_mle_vs_L} further
separates the local unbiased benchmark from constrained finite-noise risk in
these cases. In
summary, copies of one contrast do not buy more, while feature-diverse releases
stop helping once the weakest eigendirection lags $\kappa_L$.

\section{Conclusion}
\label{sec:conclusion}

This paper analyzed repeated DP aggregates through nuisance-profiled Fisher
information. The main conclusion is deliberately blunt: no new geometry, no
new target direction. Privacy accounting sets the slope of the risk curve, but
release geometry sets the identifiable span. Fixed-background \textsf{IN/OUT}
averages collapse to one profiled contrast with risk $\kappa_L/L$; Basic
Composition makes copied releases worse, while linear zCDP/RDP-style accounting
leaves them order-flat. Static permutation-invariant releases cannot identify labels,
but feature diversity can locally identify the sorted multiset until the
weakest eigendirections saturate relative to $\kappa_L$. The design rule is to
count identifiable directions, not transcripts.
\bibliographystyle{IEEEtran}
\bibliography{ref}


\clearpage


\appendices

\section{Linear-Algebra and Fisher-Information Identities}
\label{app:framework}

\subsection{Block Jacobians and Kronecker Lifting}
\label{app:block_jacobians}

For vector-valued records $\bx_k\in\R^p$, the Jacobian with respect to
record $\bx_k$ is the block
\[
\bJ_k
:=
\frac{\partial \bq(\bx_{1:N})}{\partial \bx_k^\top}
\in\R^{L\times p},
\qquad
\bJ=
\begin{bmatrix}
\bJ_1 & \cdots & \bJ_N
\end{bmatrix}.
\]
This is the block-Jacobian convention used in
Section~\ref{sec:framework} of the main text. For a coordinate-wise linear
lift of a scalar design $\bA\in\R^{L\times N}$, we stack the $p$
coordinates as separate scalar observations and write
\[
\operatorname{vec}(\by)
=
(\bA\otimes \bI_p)\operatorname{vec}(\bx_{1:N})+\bw.
\]
In the isotropic coordinate-wise noise model
$\bw\sim\mathcal N(\bm 0,\bSigma\otimes \bI_p)$, the lifted Jacobian,
whitened Jacobian, and Fisher matrix are
\[
\begin{aligned}
\bJ&=\bA\otimes \bI_p,\\
\widetilde{\bJ}&=(\bSigma^{-1/2}\bA)\otimes \bI_p,\\
\bI&=(\bA^\top\bSigma^{-1}\bA)\otimes \bI_p.
\end{aligned}
\]
More general separable coordinate covariances replace
$\bSigma\otimes \bI_p$ by $\bSigma\otimes \bSigma_0$ and lead to the
corresponding block lift. In this sense, the vector-valued labeled-target
setting of Section~\ref{sec:lever1} of the main text follows directly from
the scalar derivations. For repeated \textsf{IN/OUT} averages with isotropic
coordinate-wise noise, the trace bound is exactly $p$ times the scalar bound:
\[
\mathbb E\|\widehat{\bx}_j-\bx_j\|_2^2
\ge
p\left(
\frac{(n+1)^2\sigma_{\mathrm{in}}^2}{L_{\mathrm{in}}}
+
\frac{n^2\sigma_{\mathrm{out}}^2}{L_{\mathrm{out}}}
\right),
\]
with equality for the coordinate-wise BLUE.

\subsection{Laplace-Calibrated Location Noise}
\label{app:laplace_noise}

Independent Laplace location noise gives the same Fisher geometry after
replacing $\bSigma^{-1}$ by the diagonal Fisher weight matrix
$\bW_{\mathrm{Lap}}=\operatorname{diag}(b_1^{-2},\ldots,b_L^{-2})$.
For the standard Laplace mechanism,
$b_\ell=\Delta_{1,\ell}/\varepsilon_\ell$. Uniform budget splitting still
appends rows while reducing per-release precision. The geometry is
unchanged; only the accountant-dependent scalar weight differs.

\subsection{zCDP/RDP Scalar Calibration for Gaussian Releases}
\label{app:zcdp_rdp_calibration}

The main text writes the Gaussian noise variance in the normalized form
$\sigma_\ell^2=\kappa_L\Delta_{2,\ell}^2$. For zCDP, a scalar Gaussian
release with sensitivity $\Delta_{2,\ell}$ and variance $\sigma_\ell^2$ has
concentrated privacy parameter
\[
\rho_\ell=\frac{\Delta_{2,\ell}^2}{2\sigma_\ell^2}.
\]
Uniformly splitting a fixed total zCDP budget $\rho_{\rm tot}$ gives
$\rho_\ell=\rho_{\rm tot}/L$, and therefore
\[
\sigma_\ell^2=\frac{L}{2\rho_{\rm tot}}\Delta_{2,\ell}^2,
\qquad
\kappa_L^{\rm zCDP}=\frac{L}{2\rho_{\rm tot}}=\Theta(L).
\]
At a fixed R\'enyi order $\alpha>1$, the Gaussian release is
$(\alpha,\alpha\Delta_{2,\ell}^2/(2\sigma_\ell^2))$-RDP. If a fixed total
R\'enyi budget $\varepsilon_{\alpha,{\rm tot}}$ is split uniformly, then
\[
\sigma_\ell^2
=
\frac{\alpha L}{2\varepsilon_{\alpha,{\rm tot}}}
\Delta_{2,\ell}^2,
\qquad
\kappa_L^{\rm RDP}
=
\frac{\alpha L}{2\varepsilon_{\alpha,{\rm tot}}}
=
\Theta(L).
\]
Thus zCDP and fixed-order RDP change only the scalar accountant penalty in the
Fisher formulas. In the one-contrast \textsf{IN/OUT} case the risk is
$\kappa_L/L$, so these accountants give order-flat risk rather than the
$\Theta(L\log L)$ growth induced by Basic Composition. For feature-diverse
queries, the same substitution leaves the spectral criterion
$\lambda_{\min}(\bA_L(z))/\kappa_L$ unchanged in form.

\subsection{Projector Form and Variational Characterization of the Schur Complement}
\label{app:projector_variational}

\begin{proof}[Derivation of main Eq.~\eqref{eq:Ieff_projector}]
Since $\bI(\theta,\eta)=\widetilde{\bJ}^\top\widetilde{\bJ}$, partition
conformably as
\[
\widetilde{\bJ}
=
\begin{bmatrix}
\widetilde{\bJ}_\theta & \widetilde{\bJ}_\eta
\end{bmatrix}.
\]
Then
\[
\bI_{\theta\theta}
=
\widetilde{\bJ}_\theta^\top\widetilde{\bJ}_\theta,\qquad
\bI_{\theta\eta}
=
\widetilde{\bJ}_\theta^\top\widetilde{\bJ}_\eta,\qquad
\bI_{\eta\eta}
=
\widetilde{\bJ}_\eta^\top\widetilde{\bJ}_\eta.
\]
Substituting into main Eq.~\eqref{eq:Ieff_def} gives
\[
\bI_{\mathrm{eff}}^{(\theta)}
=
\widetilde{\bJ}_\theta^\top\widetilde{\bJ}_\theta
-
\widetilde{\bJ}_\theta^\top\widetilde{\bJ}_\eta
(\widetilde{\bJ}_\eta^\top\widetilde{\bJ}_\eta)^\dagger
\widetilde{\bJ}_\eta^\top\widetilde{\bJ}_\theta.
\]
Using the Moore--Penrose identity
$\bA(\bA^\top\bA)^\dagger\bA^\top=\bA\bA^\dagger$ with
$\bA=\widetilde{\bJ}_\eta$, we obtain
\[
\widetilde{\bJ}_\eta
(\widetilde{\bJ}_\eta^\top\widetilde{\bJ}_\eta)^\dagger
\widetilde{\bJ}_\eta^\top
=
\widetilde{\bJ}_\eta\widetilde{\bJ}_\eta^\dagger
=:\bP_\eta,
\]
which is the orthogonal projector onto
$\range(\widetilde{\bJ}_\eta)$.
Therefore,
\[
\bI_{\mathrm{eff}}^{(\theta)}
=
\widetilde{\bJ}_\theta^\top(\bI-\bP_\eta)\widetilde{\bJ}_\theta,
\]
which is main Eq.~\eqref{eq:Ieff_projector}.
\end{proof}

\paragraph*{Variational form.}
\label{prop:ieff_variational}
Write $\bA:=\bI_{\theta\theta}$, $\bC:=\bI_{\theta\eta}$, and $\bB:=\bI_{\eta\eta}$. Then
\begin{equation}
\label{eq:ieff_variational}
\bI_{\mathrm{eff}}^{(\theta)}
=
\min_{V}
\begin{bmatrix}\bI_r\\-V\end{bmatrix}^{\!\top}
\bI(\theta,\eta)
\begin{bmatrix}\bI_r\\-V\end{bmatrix}
=
\bA-\bC\bB^\dagger\bC^\top,
\end{equation}
where the minimum is in the Loewner order. Indeed,
\[
Q(V)=\bA-\bC V-V^\top\bC^\top+V^\top\bB V.
\]
Because $\bC^\top\in\range(\bB)$, $\bB\bB^\dagger\bC^\top=\bC^\top$. Completing the square with $V^\star=\bB^\dagger\bC^\top$ gives
\[
Q(V)=\bA-\bC\bB^\dagger\bC^\top+(V-V^\star)^\top\bB(V-V^\star),
\]
which proves \eqref{eq:ieff_variational}; if $\bB\succ0$, the minimizer is unique.

\subsection{Linear-Gaussian BLUE Identity}
\label{app:linear_blue_identity}

For the linear model in main Eq.~\eqref{eq:lever1_linear_model}, let
\[
\bG:=\bA^\top\bSigma^{-1}\bA.
\]
If the target functional $x_j=\be_j^\top\bx$ is identifiable, equivalently
$\be_j\in\range(\bG)$, the generalized least-squares estimator
\[
\widehat{x}_j
=
\be_j^\top \bG^\dagger \bA^\top\bSigma^{-1}\by
\]
is unbiased and has variance
\[
\Var(\widehat{x}_j)
=
\be_j^\top\bG^\dagger\be_j.
\]
Indeed, $\bG^\dagger\bG\be_j=\be_j$ for identifiable $\be_j$, and the
variance calculation follows from
$\bA^\top\bSigma^{-1}\bSigma\bSigma^{-1}\bA=\bG$ and the Moore--Penrose
identity $\bG^\dagger\bG\bG^\dagger=\bG^\dagger$. Since the Gaussian
linear experiment has Fisher matrix $\bG$, this variance equals the
Cram\'er--Rao bound for the identifiable target functional, giving main
Eq.~\eqref{eq:mse_optimal_s4}.

\subsection{Proof of Theorem~\ref{thm:spine} in the main text: Incremental Fisher Update}
\label{app:spine_rankone}

The proof has two parts: monotonicity for an added independent release,
followed by the exact rank-one update when the nuisance block is
invertible.

Let the stage-$L$ Fisher information matrix be partitioned as
\[
\bI^{(L)}
=
\begin{pmatrix}
\bA & \bC\\
\bC^\top & \bB
\end{pmatrix}.
\]
Here
\[
\bA:=\bI_{\theta\theta}^{(L)},\qquad
\bC:=\bI_{\theta\eta}^{(L)},\qquad
\bB:=\bI_{\eta\eta}^{(L)},
\]
and define
\[
Q_L(V)
:=
\begin{bmatrix}
\bI_r\\
-V
\end{bmatrix}^{\!\top}
\bI^{(L)}
\begin{bmatrix}
\bI_r\\
-V
\end{bmatrix}.
\]

Now append one more whitened observation row
\[
\begin{aligned}
\tilde{\bg}_{L+1}
&=
\begin{bmatrix}
a^\top & b^\top
\end{bmatrix},
\\
a:=a_{L+1}\in\R^r,
\qquad
b:=b_{L+1}\in\R^{Np-r}.
\end{aligned}
\]
Since Fisher information is additive for independent Gaussian
observations,
\[
\bI^{(L+1)}
=
\bI^{(L)}+\tilde{\bg}_{L+1}^\top\tilde{\bg}_{L+1}
=
\begin{pmatrix}
\bA+aa^\top & \bC+ab^\top\\
\bC^\top+ba^\top & \bB+bb^\top
\end{pmatrix}.
\]
For every $V\in\R^{(Np-r)\times r}$,
\[
\begin{aligned}
Q_{L+1}(V)
&:=
\begin{bmatrix}
\bI_r\\
-V
\end{bmatrix}^{\!\top}
\bI^{(L+1)}
\begin{bmatrix}
\bI_r\\
-V
\end{bmatrix}
\\
&=
Q_L(V)
+
(a-V^\top b)(a-V^\top b)^\top
\succeq
Q_L(V).
\end{aligned}
\]
By the variational identity~\eqref{eq:ieff_variational}, the effective Fisher
information is the minimum of $Q_L(V)$ in the Loewner order. Let
$V_L^\star$ and $V_{L+1}^\star$ be minimizers for the stage-$L$ and
stage-$(L+1)$ variational problems. Then
\[
\bI_{\mathrm{eff}}^{(\theta),L+1}
=
Q_{L+1}(V_{L+1}^\star)
\succeq
Q_L(V_{L+1}^\star)
\succeq
Q_L(V_L^\star)
=
\bI_{\mathrm{eff}}^{(\theta),L}.
\]
The second inequality uses the Loewner minimality of $V_L^\star$ for
the stage-$L$ problem. Thus
\[
\bI_{\mathrm{eff}}^{(\theta),L+1}
\succeq
\bI_{\mathrm{eff}}^{(\theta),L},
\]
which proves monotonicity without any invertibility assumption on
$\bB$.

We now turn to the exact rank-one update in the case $\bB\succ 0$,
where
\[
\bI_{\mathrm{eff}}^{(\theta),L}
=
\bA-\bC\bB^{-1}\bC^\top.
\]
Therefore
\begin{equation}
\label{eq:app_schur_Lplus1}
\bI_{\mathrm{eff}}^{(\theta),L+1}
=
(\bA+aa^\top)
-
(\bC+ab^\top)(\bB+bb^\top)^{-1}(\bC^\top+ba^\top).
\end{equation}

Applying the Sherman--Morrison identity gives
\[
(\bB+bb^\top)^{-1}
=
\bB^{-1}
-
\frac{\bB^{-1}bb^\top\bB^{-1}}{1+b^\top\bB^{-1}b}.
\]
Define
\[
\bu:=\bC\bB^{-1}b,
\qquad
t:=b^\top\bB^{-1}b.
\]
Then
\[
(\bC+ab^\top)\bB^{-1}(\bC^\top+ba^\top)
=
\bC\bB^{-1}\bC^\top+\bu a^\top+a\bu^\top+t\,aa^\top,
\]
and
\[
(\bC+ab^\top)\bB^{-1}bb^\top\bB^{-1}(\bC^\top+ba^\top)
=
(\bu+t a)(\bu+t a)^\top.
\]
Substituting these identities into~\eqref{eq:app_schur_Lplus1} yields
\begin{align*}
\bI_{\mathrm{eff}}^{(\theta),L+1}
&=
\bA-\bC\bB^{-1}\bC^\top
\\
&\quad
{}+aa^\top-\bu a^\top-a\bu^\top-t\,aa^\top
\nonumber\\
&\quad
{}+
\frac{(\bu+t a)(\bu+t a)^\top}{1+t}.
\end{align*}
Combining the last two terms over the common denominator $1+t$ gives
\[
\begin{aligned}
&aa^\top-\bu a^\top-a\bu^\top-t\,aa^\top\\
&\quad{}+
\frac{(\bu+t a)(\bu+t a)^\top}{1+t}\\
&=
\frac{(a-\bu)(a-\bu)^\top}{1+t}.
\end{aligned}
\]
Hence
\begin{equation}
\label{eq:app_rankone_update}
\bI_{\mathrm{eff}}^{(\theta),L+1}
=
\bI_{\mathrm{eff}}^{(\theta),L}
+
\frac{(a-\bC\bB^{-1}b)(a-\bC\bB^{-1}b)^\top}
{1+b^\top\bB^{-1}b}.
\end{equation}
Recalling that
\[
d_{L+1}
:=
a_{L+1}
-
\bI_{\theta\eta}^{(L)}
\bigl(\bI_{\eta\eta}^{(L)}\bigr)^{-1}
b_{L+1},
\]
we obtain exactly main Eq.~\eqref{eq:spine_update}. Since the denominator is
strictly positive, equality holds if and only if the rank-one
numerator vanishes, i.e., if and only if $d_{L+1}=\bm 0$, proving
main Eq.~\eqref{eq:spine_iff}.

\section{Labeled-Target Reconstruction}
\label{app:lever1}

\subsection{One-Contrast Gram Geometry}
\label{app:one_contrast_gram}

\paragraph{Release-level routes to nuisance permutation symmetry}
The Gram-level symmetry in the main text is obtained whenever the whitened
release Jacobians are invariant under relabeling the nuisance records. A
minimal sufficient condition is that, for every release row, all nuisance
coordinates in $B$ have the same derivative weight, and that any rows which
distinguish individual nuisance records are included together with their full
permutation orbit and equal noise weights. These conditions force the
restricted Gram matrix to have equal target--nuisance couplings, equal nuisance
diagonal entries, and equal nuisance off-diagonal entries, which is precisely
the block form in main Eq.~\eqref{eq:lever1_blocksym_gram}. If the orbit spans
both the nuisance-average direction and the nuisance-orthogonal directions,
then the nuisance block $\bB$ in the main lemma is nonsingular.

A degenerate but useful example is static membership $C_\ell\equiv C$ with a
common feature map $\phi_\ell\equiv\phi$ and a common operating-point derivative
$J_\phi(\bx_k)=J_\phi^\star$ for every $k\in C$. Then main
Eq.~\eqref{eq:deepsets_grad} reduces to
\[
\frac{\partial q_\ell(\bx_{1:N})}{\partial x_k}
=
J_{\rho_\ell}(s_\ell)\,J_\phi^\star,
\qquad k\in C.
\]
All restricted Jacobian columns indexed by $k\in C$ are identical, so the
Gram matrix is invariant under nuisance permutations. In this degenerate case
the nuisance block may be singular; it illustrates the symmetry mechanism but
not the positive-definite nuisance-block condition.

\paragraph{Permutation-invariant Gram form on the nuisance block}
Permutation invariance among the nuisance coordinates implies three equalities:
all entries of the target-to-nuisance block are equal, all diagonal entries
of the nuisance block are equal, and all off-diagonal entries of the
nuisance block are equal. These constraints yield the block form in main
Eq.~\eqref{eq:lever1_blocksym_gram}.

\begin{proof}[Proof of Lemma~\ref{thm:lever1_onecontrast} in the main text]
Write
\[
\bG_C
=
\begin{bmatrix}
a & c\,\mathbf 1^\top\\
c\,\mathbf 1 & \bB
\end{bmatrix},
\qquad
\bB=(d-e)\bI_{m-1}+e\,\mathbf 1\mathbf 1^\top.
\]
Since $\bB\succ 0$, the Schur complement gives
\[
\bI_{\mathrm{eff}}^{(j)}
=
a-c^2\,\mathbf 1^\top \bB^{-1}\mathbf 1.
\]
Using the rank-one inversion formula,
\[
\bB^{-1}
=
\frac{1}{d-e}\bI_{m-1}
-
\frac{e}{(d-e)(d+(m-2)e)}\mathbf 1\mathbf 1^\top,
\]
and therefore
\[
\mathbf 1^\top \bB^{-1}\mathbf 1
=
\frac{m-1}{d+(m-2)e}.
\]
Substituting this expression into the Schur complement yields main
Eq.~\eqref{eq:lever1_onecontrast_ieff}. The target-to-nuisance coupling is
proportional to $\mathbf 1$, so nuisance directions orthogonal to
$\mathbf 1$ do not enter the Schur complement.
\end{proof}

\subsection{Reduction to the Two-Parameter \textsf{IN/OUT} Experiment}
\label{app:lever1_fixed_background_limit}

\begin{proof}[Reduction to two parameters]
In the repeated \textsf{IN/OUT} family of Section~\ref{sec:lever1}, every
release belongs to one of two types. For an \textsf{IN} release,
\[
\begin{aligned}
y_\ell
&=
\frac{1}{n+1}\Big(x_j+\sum_{k\in B}x_k\Big)+w_\ell
\\
&=
\frac{1}{n+1}(\theta_1+\theta_2)+w_\ell,
\qquad \ell\in\mathcal I_{\mathrm{in}},
\end{aligned}
\]
whereas for an \textsf{OUT} release,
\[
\begin{aligned}
y_\ell
&=
\frac{1}{n}\sum_{k\in B}x_k+w_\ell
\\
&=
\frac{1}{n}\theta_2+w_\ell,
\qquad \ell\in\mathcal I_{\mathrm{out}}.
\end{aligned}
\]
Hence the full observation vector $\by$ depends on $\bx$ only through
$(\theta_1,\theta_2)$, which proves the first claim. Since
$\theta_1\equiv x_j$, target reconstruction in the full experiment is
exactly reconstruction of the first coordinate in this reduced
experiment. The two reduced mean rows are
$((n+1)^{-1},(n+1)^{-1})$ and $(0,n^{-1})$, whose determinant is
$1/(n(n+1))\neq0$. Thus the reduced experiment is locally identifiable
whenever both an \textsf{IN} and an \textsf{OUT} block are present. If
$L_{\mathrm{out}}=0$, the data contain only the combination
$\theta_1+\theta_2$; if $L_{\mathrm{in}}=0$, the target coordinate is
absent. These are the confounding cases excluded in the main text.

It remains to check that replacing each block by its empirical mean
preserves the Fisher information for estimating $\theta_1$ in the
presence of the nuisance parameter $\theta_2$.
Conditioned on $(\theta_1,\theta_2)$, the \textsf{IN} observations are independent
Gaussians with common mean
\[
\mu_{\mathrm{in}}(\theta_1,\theta_2)=\frac{1}{n+1}(\theta_1+\theta_2)
\]
and common variance $\sigma_{\mathrm{in}}^2$, while the \textsf{OUT} observations are
independent Gaussians with common mean
\[
\mu_{\mathrm{out}}(\theta_2)=\frac{1}{n}\theta_2
\]
and common variance $\sigma_{\mathrm{out}}^2$.
Let
$\varphi_\sigma(t):=(2\pi\sigma^2)^{-1/2}\exp\{-t^2/(2\sigma^2)\}$.
The joint likelihood therefore factors as
\[
\begin{aligned}
p(\by\mid \theta_1,\theta_2)
&=
\prod_{\ell\in\mathcal I_{\mathrm{in}}}
\varphi_{\sigma_{\mathrm{in}}}\!\left(y_\ell-\mu_{\mathrm{in}}\right)
\\
&\quad{}\times
\prod_{\ell\in\mathcal I_{\mathrm{out}}}
\varphi_{\sigma_{\mathrm{out}}}\!\left(y_\ell-\mu_{\mathrm{out}}\right).
\end{aligned}
\]
Using the identity
\[
\sum_{\ell\in\mathcal I_{\mathrm{in}}}(y_\ell-\mu)^2
=
\sum_{\ell\in\mathcal I_{\mathrm{in}}}(y_\ell-\bar y_{\mathrm{in}})^2
+
L_{\mathrm{in}}(\bar y_{\mathrm{in}}-\mu)^2,
\]
and the analogous decomposition for the \textsf{OUT} block, we obtain the
factorization
\[
p(\by\mid \theta_1,\theta_2)
=
h(\by)\,
g(\bar y_{\mathrm{in}},\bar y_{\mathrm{out}};\theta_1,\theta_2),
\]
where $h(\by)$ is free of $(\theta_1,\theta_2)$.
By the Fisher--Neyman factorization theorem,
$(\bar y_{\mathrm{in}},\bar y_{\mathrm{out}})$ is sufficient for
$(\theta_1,\theta_2)$.

Moreover, since the two blocks are independent Gaussians, the compressed pair
satisfies
\[
\bar y_{\mathrm{in}}
=
\frac{1}{n+1}(\theta_1+\theta_2)+\bar w_{\mathrm{in}},
\qquad
\bar w_{\mathrm{in}}\sim
\mathcal N\!\Big(0,\frac{\sigma_{\mathrm{in}}^2}{L_{\mathrm{in}}}\Big),
\]
and
\[
\bar y_{\mathrm{out}}
=
\frac{1}{n}\theta_2+\bar w_{\mathrm{out}},
\qquad
\bar w_{\mathrm{out}}\sim
\mathcal N\!\Big(0,\frac{\sigma_{\mathrm{out}}^2}{L_{\mathrm{out}}}\Big),
\]
with $\bar w_{\mathrm{in}}$ and $\bar w_{\mathrm{out}}$ independent. The
reduced $2\times2$ design matrix has rank two exactly when
$L_{\mathrm{in}},L_{\mathrm{out}}>0$; otherwise $\theta_1=x_j$ is confounded
with the background sum $\theta_2$. This is the identifiability condition
used in Proposition~\ref{prop:fixed_background_crb} of the main text.
The Fisher information matrix of the full sample is the sum of the
per-sample Fisher information matrices:
\[
\bI_{\mathrm{full}}
=
\frac{L_{\mathrm{in}}}{\sigma_{\mathrm{in}}^2}
\begin{bmatrix}
\frac{1}{(n+1)^2} & \frac{1}{(n+1)^2}\\[3pt]
\frac{1}{(n+1)^2} & \frac{1}{(n+1)^2}
\end{bmatrix}
+
\frac{L_{\mathrm{out}}}{\sigma_{\mathrm{out}}^2}
\begin{bmatrix}
0 & 0\\[3pt]
0 & \frac{1}{n^2}
\end{bmatrix}.
\]
The Fisher information matrix of the compressed model
$(\bar y_{\mathrm{in}},\bar y_{\mathrm{out}})$ is exactly the same, because
compressing $L_{\mathrm{in}}$ i.i.d.\ Gaussian samples with common mean into
$\bar y_{\mathrm{in}}$ simply replaces variance $\sigma_{\mathrm{in}}^2$ by
$\sigma_{\mathrm{in}}^2/L_{\mathrm{in}}$, and likewise for the \textsf{OUT} block.
Hence the Fisher information for estimating $\theta_1$ in the presence of
nuisance $\theta_2$ is preserved by the block means.
\end{proof}

\subsection{Exact CRB and BLUE Attainability}
\label{app:fixed_background_crb_blue}

\begin{proof}[Proof of Proposition~\ref{prop:fixed_background_crb} in the main text]
Consider the reduced model
\[
\bar y_{\mathrm{in}}
=
\frac{1}{n+1}(\theta_1+\theta_2)+\bar w_{\mathrm{in}},
\qquad
\bar y_{\mathrm{out}}
=
\frac{1}{n}\theta_2+\bar w_{\mathrm{out}},
\]
where
\[
\bar w_{\mathrm{in}}\sim
\mathcal N\!\Big(0,\frac{\sigma_{\mathrm{in}}^2}{L_{\mathrm{in}}}\Big),
\qquad
\bar w_{\mathrm{out}}\sim
\mathcal N\!\Big(0,\frac{\sigma_{\mathrm{out}}^2}{L_{\mathrm{out}}}\Big),
\]
and the two noises are independent.
Writing
\[
\bar{\by}
:=
\begin{bmatrix}
\bar y_{\mathrm{in}}\\
\bar y_{\mathrm{out}}
\end{bmatrix},
\qquad
\btheta
:=
\begin{bmatrix}
\theta_1\\
\theta_2
\end{bmatrix},
\]
the model is
\[
\begin{aligned}
\bar{\by}&=\bM \btheta + \bar{\bw},\\
\bar{\bw}&\sim \mathcal N(0,\bSigma_{\rm red}),
\end{aligned}
\]
where
\[
\bM
=
\begin{bmatrix}
\frac{1}{n+1} & \frac{1}{n+1}\\[4pt]
0 & \frac{1}{n}
\end{bmatrix},
\]
with
\[
\bSigma_{\rm red}
=
\begin{bmatrix}
\sigma_{\mathrm{in}}^2/L_{\mathrm{in}} & 0\\[3pt]
0 & \sigma_{\mathrm{out}}^2/L_{\mathrm{out}}
\end{bmatrix}.
\]

The Fisher information matrix for $(\theta_1,\theta_2)$ is therefore
\[
\bI_{\rm red}
=
\bM^\top \bSigma_{\rm red}^{-1}\bM.
\]
Define
\[
a:=\frac{L_{\mathrm{in}}}{(n+1)^2\sigma_{\mathrm{in}}^2},
\qquad
b:=\frac{L_{\mathrm{out}}}{n^2\sigma_{\mathrm{out}}^2}.
\]
A direct calculation gives
\[
\bI_{\rm red}
=
\begin{bmatrix}
a & a\\
a & a+b
\end{bmatrix}.
\]

To estimate the target parameter $\theta_1$ in the presence of nuisance $\theta_2$,
the effective Fisher information is the Schur complement of the $(2,2)$ entry:
\[
I_{\mathrm{eff}}^{(\theta_1)}
=
a-\frac{a^2}{a+b}
=
\frac{ab}{a+b}.
\]
Hence the nuisance-aware Cram\'er--Rao bound is
\[
\mathrm{MSE}^*(\theta_1)
=
\frac{1}{I_{\mathrm{eff}}^{(\theta_1)}}
=
\frac{a+b}{ab}
=
\frac{1}{a}+\frac{1}{b}.
\]
Substituting the definitions of $a$ and $b$ yields
\[
\mathrm{MSE}^*(x_j)
=
\frac{(n+1)^2\sigma_{\mathrm{in}}^2}{L_{\mathrm{in}}}
+
\frac{n^2\sigma_{\mathrm{out}}^2}{L_{\mathrm{out}}},
\]
which proves main Eq.~\eqref{eq:crb_fixed_background}.

It remains to show attainability and identify the BLUE.
Consider an estimator of the form
\[
\hat\theta_1=\alpha \bar y_{\mathrm{in}}+\beta \bar y_{\mathrm{out}}.
\]
For unbiasedness for all $(\theta_1,\theta_2)$, we require
\[
\mathbb E[\hat\theta_1]
=
\alpha\frac{1}{n+1}(\theta_1+\theta_2)+\beta\frac{1}{n}\theta_2
=
\theta_1.
\]
Matching coefficients of $\theta_1$ and $\theta_2$ gives
\[
\frac{\alpha}{n+1}=1,
\qquad
\frac{\alpha}{n+1}+\frac{\beta}{n}=0,
\]
hence
\[
\alpha=n+1,
\qquad
\beta=-n.
\]
Therefore the unique unbiased linear estimator is
\[
\hat x_j
=
(n+1)\bar y_{\mathrm{in}}-n\bar y_{\mathrm{out}}.
\]
Its variance is
\[
\begin{aligned}
\Var(\hat x_j)
&=
(n+1)^2\Var(\bar y_{\mathrm{in}})
+n^2\Var(\bar y_{\mathrm{out}})
\\
&=
\frac{(n+1)^2\sigma_{\mathrm{in}}^2}{L_{\mathrm{in}}}
+
\frac{n^2\sigma_{\mathrm{out}}^2}{L_{\mathrm{out}}},
\end{aligned}
\]
since the two block means are independent.
This variance equals the CRB, so the estimator is BLUE and attains the
minimum achievable MSE.

For the continuous relaxation with
\[
L_{\mathrm{in}},L_{\mathrm{out}}>0,
\qquad
L_{\mathrm{in}}+L_{\mathrm{out}}=L,
\]
write
\[
\alpha:=(n+1)^2\sigma_{\mathrm{in}}^2,
\qquad
\beta:=n^2\sigma_{\mathrm{out}}^2.
\]
Then
\[
\mathrm{MSE}^*(x_j)
=
\frac{\alpha}{L_{\mathrm{in}}}
+
\frac{\beta}{L_{\mathrm{out}}}.
\]
By the Cauchy--Schwarz inequality,
\[
\left(
\frac{\alpha}{L_{\mathrm{in}}}
+
\frac{\beta}{L_{\mathrm{out}}}
\right)
(L_{\mathrm{in}}+L_{\mathrm{out}})
\ge
(\sqrt{\alpha}+\sqrt{\beta})^2,
\]
with equality if and only if
\[
\frac{L_{\mathrm{in}}}{L_{\mathrm{out}}}
=
\frac{\sqrt{\alpha}}{\sqrt{\beta}}
=
\frac{n+1}{n}\,\frac{\sigma_{\mathrm{in}}}{\sigma_{\mathrm{out}}}.
\]
Since $L_{\mathrm{in}}+L_{\mathrm{out}}=L$, this proves
main Eq.~\eqref{eq:opt_split_s4} and yields
\[
\begin{aligned}
\min_{\substack{L_{\mathrm{in}},L_{\mathrm{out}}>0\\
L_{\mathrm{in}}+L_{\mathrm{out}}=L}}
\mathrm{MSE}^*(x_j)
&=
\frac{(\sqrt{\alpha}+\sqrt{\beta})^2}{L}
\\
&=
\frac{1}{L}\Bigl((n+1)\sigma_{\mathrm{in}}+n\sigma_{\mathrm{out}}\Bigr)^2,
\end{aligned}
\]
which is main Eq.~\eqref{eq:mse_inout_opt_s4}. Because the objective is strictly
convex in $L_{\mathrm{in}}\in(0,L)$, the optimal integer split is
obtained by the nearest feasible rounding around this continuous
optimum.
\end{proof}

\subsection{Non-Uniform Allocation}
\label{app:lever1_nonuniform}

\begin{proof}[Proof of Proposition~\ref{prop:lever1_nonuniform} in the main text]
For subset averages,
\[
\Delta_{2,\mathrm{in}}=\frac{\Delta}{n+1},
\qquad
\Delta_{2,\mathrm{out}}=\frac{\Delta}{n}.
\]
Substituting the Gaussian calibration into main Eq.~\eqref{eq:crb_fixed_background}
gives main Eq.~\eqref{eq:mse_inout_eps_delta}; the factors $(n+1)^2$ and $n^2$ from the
BLUE coefficients cancel the corresponding sensitivities.

For fixed $(L_{\mathrm{in}},L_{\mathrm{out}})$ and fixed $(\delta_{\mathrm{in}},\delta_{\mathrm{out}})$, set
$c_{\mathrm{in}}:=\ln(1.25/\delta_{\mathrm{in}})$,
$c_{\mathrm{out}}:=\ln(1.25/\delta_{\mathrm{out}})$,
$u:=L_{\mathrm{in}}\varepsilon_{\mathrm{in}}$, and $v:=L_{\mathrm{out}}\varepsilon_{\mathrm{out}}$.
Then $u+v=\varepsilon_{\mathrm{tot}}$ and
\begin{equation}
\label{eq:mse_inout_opt_eps}
\min_{\varepsilon}\mathrm{MSE}^*(x_j)
=
\frac{2\Delta^2}{\varepsilon_{\mathrm{tot}}^{\,2}}
\Big(
c_{\mathrm{in}}^{1/3}L_{\mathrm{in}}^{1/3}
+
c_{\mathrm{out}}^{1/3}L_{\mathrm{out}}^{1/3}
\Big)^3.
\end{equation}
Indeed,
\[
\mathrm{MSE}^*(x_j)
=
2\Delta^2\left(\frac{c_{\mathrm{in}}L_{\mathrm{in}}}{u^2}+\frac{c_{\mathrm{out}}L_{\mathrm{out}}}{v^2}\right),
\]
and the KKT condition
\[
-\frac{2c_{\mathrm{in}}L_{\mathrm{in}}}{u^3}
=
-\frac{2c_{\mathrm{out}}L_{\mathrm{out}}}{v^3}
\]
gives the cube-root rule in main Eq.~\eqref{eq:opt_u_v_ratio}.
Under symmetric calibration $c_{\mathrm{in}}=c_{\mathrm{out}}=:c$, \eqref{eq:mse_inout_opt_eps} reduces to
main Eq.~\eqref{eq:mse_inout_opt_eps_symmetric}. For comparing release counts
under a fixed total failure budget $\delta_{\mathrm{tot}}$, take the
feasible symmetric allocation
\[
\delta_{\mathrm{in}}=\delta_{\mathrm{out}}
=\frac{\delta_{\mathrm{tot}}}{L_{\mathrm{in}}+L_{\mathrm{out}}},
\]
so
\[
c=c(L_{\mathrm{in}},L_{\mathrm{out}})
=
\ln\!\left(
\frac{1.25(L_{\mathrm{in}}+L_{\mathrm{out}})}{\delta_{\mathrm{tot}}}
\right).
\]
Both $c(L_{\mathrm{in}},L_{\mathrm{out}})$ and
$(L_{\mathrm{in}}^{1/3}+L_{\mathrm{out}}^{1/3})^3$ are nondecreasing in
each release count for $L_{\mathrm{in}},L_{\mathrm{out}}\ge 1$.
Therefore the fixed-$\delta_{\mathrm{tot}}$ comparison is minimized at
$L_{\mathrm{in}}=L_{\mathrm{out}}=1$, which proves the final claim of
Proposition~\ref{prop:lever1_nonuniform} in the main text.
\end{proof}

\paragraph{Cube-root allocation calculation}
\label{app:aux_opt}
The same rule follows from the two-variable Lagrangian
\[
\mathcal L(u,v,\lambda)
=
\frac{c_{\mathrm{in}}L_{\mathrm{in}}}{u^2}
+
\frac{c_{\mathrm{out}}L_{\mathrm{out}}}{v^2}
+
\lambda(u+v-\varepsilon_{\mathrm{tot}}).
\]
The first-order conditions give
$c_{\mathrm{in}}L_{\mathrm{in}}/u^3=c_{\mathrm{out}}L_{\mathrm{out}}/v^3$,
which is exactly the cube-root ratio in main Eq.~\eqref{eq:opt_u_v_ratio}; the
constraint $u+v=\varepsilon_{\mathrm{tot}}$ then determines the unique split.

\section{Unlabeled Quotient Geometry}
\label{app:lever2}

\subsection{Proof of the Local Quotient Chart}
\label{app:sorted_chart}

\begin{proof}[Proof of Proposition~\ref{prop:lever2_sorted_chart_main} in the main text]
Assume the points $\{x_k\}_{k\in C}$ are pairwise distinct, and let
$\bx_C^\uparrow:=(x_{(1)},\dots,x_{(m)})$ denote the sorted
representative. On any neighborhood that avoids the collision hyperplanes
$x_a=x_b$, the permutation that sorts the coordinates is locally constant. Hence the sorting map is locally linear and equal to
$x\mapsto P^\top x$ for some fixed permutation matrix $P$ on that
neighborhood, which proves the existence of the local quotient chart.
The local unlabeled matching problem is therefore attained by comparing
sorted representatives; this gives the local equivalence of the two loss
functions. Under the reparameterization
$\bx_C^\uparrow=P^\top\bx_C$, the Jacobian transforms as
$\bJ_{\rm sym}=\bJ_C P$, and hence
\[
\begin{aligned}
\bI_{\rm sym}(L)
&=
\bJ_{\rm sym}^\top \bSigma_L^{-1}\bJ_{\rm sym}
\\
&=
P^\top \bJ_C^\top \bSigma_L^{-1}\bJ_C P
\\
&=
P^\top \bI_C(L)P.
\end{aligned}
\]
Since trace is invariant under permutation similarity, the benchmark
definition in main Eq.~\eqref{eq:crb_unlab_def} is well posed in sorted
coordinates.
\end{proof}

\subsection{Non-Identifiability under Static-Membership Inclusion}
\label{app:static_membership}

\paragraph*{Static-membership consequence.}
\label{cor:membership_necessity_s3}
Partition $\widetilde{\bJ}=[\widetilde{\bJ}_j\ \widetilde{\bJ}_{-j}]$. If $\mathrm{range}(\widetilde{\bJ}_j)\subseteq \mathrm{range}(\widetilde{\bJ}_{-j})$, then $\bI_{\mathrm{eff}}^{(j)}=\bm0$ and the corresponding CRB diverges. To see this, let $\bP_{-j}:=\widetilde{\bJ}_{-j}\widetilde{\bJ}_{-j}^\dagger$. The range inclusion implies $(\bI-\bP_{-j})\widetilde{\bJ}_j=\bm0$, and main Eq.~\eqref{eq:Ieff_projector} gives the claim.

\section{Query Families}
\label{app:query_families}

\subsection{Polynomial moments: sensitivity and weighted Vandermonde rank}
\label{app:poly}

This appendix records the intermediate calculations suppressed in
Section~\ref{subsec:lever2_poly} of the main text. For
\[
q_\ell(z)=\frac{1}{m}\sum_{k=1}^m z_k^\ell,
\qquad \ell=1,\dots,L,
\]
the Jacobian entries are
\begin{equation}
\label{eq:app_poly_jacobian}
(\bJ_C)_{\ell k}
=
\frac{\ell}{m}z_k^{\ell-1}.
\end{equation}
Under fixed-cardinality swap adjacency with $|z_k|\le R$, changing one record
can alter the released average by at most
\begin{equation}
\label{eq:app_poly_sensitivity}
\Delta_{2,\ell}^{\rm swap}\le \frac{2R^\ell}{m}.
\end{equation}

\paragraph{Weighted Vandermonde rank fact}
Write $\bJ_C=\frac{1}{m}D\,V$, where $D=\mathrm{diag}(1,\dots,L)$ and
$V_{\ell k}=z_k^{\ell-1}$ is a Vandermonde matrix. Distinct nodes imply
$\mathrm{rank}(V)=\min\{L,m\}$, hence
$\mathrm{rank}(\bJ_C)=\min\{L,m\}$. Whitening by $\bSigma_L^{-1/2}$
preserves rank because it is invertible.

Under uniform Basic Composition,
\begin{equation}
\label{eq:app_poly_sigma}
\sigma_\ell^2
=
\kappa_L\left(\Delta_{2,\ell}^{\rm swap}\right)^2
=
\kappa_L\,\frac{4R^{2\ell}}{m^2}.
\end{equation}

Writing $r_k:=z_k/R$ and $\bm{\nu}_\ell:=(r_k^{\ell-1})_{k=1}^m$, we obtain
\begin{align}
\bI_C^{\mathrm{poly}}(L)
&=
\sum_{\ell=1}^L \frac{1}{\sigma_\ell^2}\,\bJ_{\ell,C}^\top\bJ_{\ell,C}
\nonumber\\
&=
\sum_{\ell=1}^L
\frac{m^2}{4\kappa_L R^{2\ell}}
\left(\frac{\ell}{m}z_k^{\ell-1}\right)_{k\in C}
\left(\frac{\ell}{m}z_k^{\ell-1}\right)_{k\in C}^{\!\top}
\nonumber\\
&=
\frac{1}{4R^2\kappa_L}
\sum_{\ell=1}^L \ell^2\bm{\nu}_\ell\bm{\nu}_\ell^\top,
\label{eq:app_poly_rankone}
\end{align}
which matches main Eq.~\eqref{eq:I_poly_rankone_new}. The averaging factor $1/m$ cancels
exactly against the sensitivity scaling, so the remaining competition is purely
between the signal sum and the common composition factor $\kappa_L$.

For a sorted coordinate $z_j$, write $r_j:=z_j/R$. Then the normalized
cumulative signal satisfies
\[
e_j^\top \bA_L(z)e_j
=
\frac{1}{4R^2}\sum_{\ell=1}^L \ell^2 r_j^{2\ell-2}
=
\frac{1}{4R^2}T_j(L),
\]
which is the coordinate-level form used in
Section~\ref{subsec:lever2_poly} of the main text. The corresponding asymptotic implications for
interior and boundary points are developed in \appref{app:query_families}{app:m1_additional}.

\subsection{Polynomial moments: coordinate-wise bounds and stopping rule}
\label{app:m1_additional}
For a sorted coordinate $z_j$ define
\begin{equation}
\label{eq:T_def}
r:=z_j/R,
\qquad
T(L):=\sum_{\ell=1}^L \ell^2 r^{2(\ell-1)} .
\end{equation}
Then $(\bI_C(L))_{jj}=(4\kappa_LR^2)^{-1}T(L)$. If $\bI_C(L)$ is nonsingular, the Schur-complement formula for the $(j,j)$ entry of $\bI_C(L)^{-1}$ gives
\begin{equation}
\label{eq:unlab_crlb_kappa_over_T}
\mathrm{CRB}_{\rm unlab}(L)
\ge (\bI_C(L)^{-1})_{jj}
\ge 4\kappa_LR^2/T(L).
\end{equation}
For $|r|<1$,
\begin{equation}
\label{eq:T_limit}
T(L)\le T(\infty)=\frac{1+r^2}{(1-r^2)^3}<\infty,
\end{equation}
so \eqref{eq:unlab_crlb_kappa_over_T} diverges with $\kappa_L$. If $|r|=1$, then $T(L)=\sum_{\ell=1}^L\ell^2=\Theta(L^3)$ and the coordinate information scales as $T(L)/\kappa_L=\Theta(L/\log L)$ under Basic Composition. In the scalar case, the exact consecutive-risk comparison is
\begin{equation}
\label{eq:m1_stop_appendix}
\frac{\mathrm{CRB}_{m=1}(L+1)}{\mathrm{CRB}_{m=1}(L)}
=
\frac{\kappa_{L+1}}{\kappa_L}
\frac{T(L)}{T(L)+(L+1)^2r^{2L}},
\end{equation}
which gives the stopping rule used to locate the finite interior optimum. For
$q=r^2\in(0,1)$, the exact tail is
\[
T(\infty)-T(L)
=
\frac{q^L\{(L+1)^2-(2L^2+2L-1)q+L^2q^2\}}{(1-q)^3},
\]
so the marginal gain from additional moments is geometrically damped. The useful moment orders are exhausted once $L\log(1/q)$ dominates the
polynomial factor in $L$ and the slowly varying logarithm in $\kappa_L$, giving
the $1/\log(1/q)$ stopping scale quoted in the main text.

\subsection{Polynomial moments: closed-form \texorpdfstring{$m=2$}{m=2} case}
\label{app:m2}

This appendix supplies the short calculations underlying the closed-form
$m=2$ illustration in Section~\ref{subsec:lever2_m2} of the main text.

\paragraph{Degeneracy at $L=1$}
From main Eq.~\eqref{eq:m2_fim}, the determinant is proportional to
\[
D(L):=F(r_1^2,L)F(r_2^2,L)-F(\rho,L)^2.
\]
Since $F(s,1)=1$ for every $s$, we have $D(1)=0$, so the one-moment Fisher
matrix is singular and the unlabeled CRB diverges.

\paragraph{Two-moment simplification}
At $L=2$, the identity $F(s,2)=1+4s$ gives
\begin{align*}
D(2)
&=
\left(1+\frac{4x_1^2}{R^2}\right)
\left(1+\frac{4x_2^2}{R^2}\right)
-\left(1+\frac{4x_1x_2}{R^2}\right)^2\\
&=
\frac{4(x_2-x_1)^2}{R^2}
=
\frac{16\gamma^2}{R^2},
\end{align*}
which is the denominator used to obtain main Eq.~\eqref{eq:m2_L2}. The same calculation
also makes the separation penalty explicit: the CRB blows up as
$\gamma\downarrow 0$.

\paragraph{Canonical estimator}
For $L=2$, the noiseless moment equations
\[
y_1=\frac{x_1+x_2}{2},
\qquad
y_2=\frac{x_1^2+x_2^2}{2}
\]
imply
\[
x_1x_2=2y_1^2-y_2,
\qquad
(x_2-x_1)^2=4(y_2-y_1^2),
\]
which leads directly to the algebraic inversion formula
\begin{equation}
\label{eq:m2_estimator}
\hat x_{(1)}=y_1-\sqrt{(y_2-y_1^2)_+},
\qquad
\hat x_{(2)}=y_1+\sqrt{(y_2-y_1^2)_+}.
\end{equation}
For general $L$, the nonlinear least-squares objective
used in Section~\ref{sec:experiments} of the main text coincides with maximum likelihood under
Gaussian noise, so the same estimator family is the natural numerical benchmark
for testing CRB tightness.

\subsection{Polynomial moments: general-\texorpdfstring{$m$}{m} determinant bound}
\label{app:poly_general_m}

\paragraph*{Determinant/AM--GM bound.}
\label{prop:lever2_det_amgm}
Define $\bW\in\R^{L\times m}$ by $W_{\ell k}:=\ell z_k^{\ell-1}/R^\ell$. Then $\bI_C(L)=(4\kappa_L)^{-1}\bW^\top\bW$. If $L\ge m$ and the points are pairwise distinct, Hadamard's inequality gives
\begin{equation}
\label{eq:det_bound_general_m}
\det(\bI_C(L))
\le
\frac{1}{(4\kappa_L R^2)^m}
\prod_{k=1}^m F(r_k^2,L),
\qquad
r_k:=z_k/R,
\end{equation}
where $F(s,L):=\sum_{\ell=1}^L\ell^2s^{\ell-1}$. If $\lambda_i>0$ are the eigenvalues of $\bI_C(L)$, AM--GM applied to $\lambda_i^{-1}$ yields
\begin{equation}
\label{eq:crb_det_amgm_bound}
\begin{aligned}
\mathrm{CRB}_{\rm unlab}(L)
&=\sum_{i=1}^m\lambda_i^{-1}
\ge
\frac{m}{\det(\bI_C(L))^{1/m}}\\
&\ge
\frac{4m\kappa_LR^2}
{\bigl(\prod_{k=1}^m F(r_k^2,L)\bigr)^{1/m}}.
\end{aligned}
\end{equation}
\paragraph*{General-$m$ interior consequence.}
\label{cor:lever2_generalm_interior}
If $|r_k|<1$ for every $k$, then $F(r_k^2,L)\le(1+r_k^2)/(1-r_k^2)^3<\infty$, while $\kappa_L=\Theta(L^2\log L)$ under Basic Composition. Thus \eqref{eq:crb_det_amgm_bound} diverges as $L\to\infty$; combined with rank deficiency for $L<m$, the interior local risk has a finite minimizer.

\paragraph*{Scalar stopping scale.}
For $m=1$, write $s=r^2<1$. The closed form for the tail of
$T_1(L)=\sum_{\ell=1}^L\ell^2s^{\ell-1}$ is
\[
T_1(\infty)-T_1(L)
=
\frac{s^L\{(L+1)^2-(2L^2+2L-1)s+L^2s^2\}}{(1-s)^3}.
\]
The one-step condition for decreasing local risk is
$T_1(L+1)/T_1(L)>\kappa_{L+1}/\kappa_L$. Because the increment
$T_1(L+1)-T_1(L)=(L+1)^2s^L$ is geometrically damped, the useful moment
orders are exhausted once $L\log(1/s)$ dominates the polynomial factors in
$L$ and the slow variation of the logarithm in $\kappa_L$. This gives the
$1/\log(1/s)$ stopping scale, up to logarithmic factors, stated in the main
text.

\subsection{Threshold releases: counting argument}
\label{app:threshold}

We prove the divergence statement in Proposition~\ref{prop:threshold_truncation}(ii)
of the main text, and then record a matching dense-grid upper bound under
an explicit kernel-conditioning assumption.
Part~(i) gives
\[
\tr(\bA_L(z))=O(L).
\]
Whenever $L\ge L_{\mathrm{id}}$, we have $\bA_L(z)\succ\bm 0$, so
\[
\lambda_{\max}(\bA_L(z))
\le
\tr(\bA_L(z))
=
O(L),
\]
and therefore
\[
\tr\!\bigl(\bA_L(z)^{-1}\bigr)
\ge
\frac{1}{\lambda_{\max}(\bA_L(z))}
=
\Omega\!\left(\frac{1}{L}\right).
\]
The spectral representation in Theorem~\ref{thm:lever2_design_law} of the main
text yields
\[
\mathrm{CRB}_{\rm unlab}(L)
=
\kappa_L\,\tr\!\bigl(\bA_L(z)^{-1}\bigr)
=
\Omega\!\left(\frac{\kappa_L}{L}\right).
\]
Under uniform Basic Composition with Gaussian calibration,
$\kappa_L=\Theta(L^2\log L)$ by main Eq.~\eqref{eq:kappaL_sigma_def}, and hence
\[
\mathrm{CRB}_{\rm unlab}(L)=\Omega(L\log L)\to\infty.
\]
Since the benchmark is infinite for $L<L_{\mathrm{id}}$, finite at the first
identifiable stage, and diverges as $L\to\infty$, the minimizing release count
is finite under the fixed-budget Basic-Composition calibration.

\paragraph*{Dense-grid matching rate.}
The lower bound above is tight under the nondegenerate dense-grid conditions
used in the main-text heuristic. Suppose the threshold grid has spacing
$\eta_L\asymp L^{-1}$, every $z_k$ remains a fixed positive distance from the
grid endpoints, and the continuous kernel Gram matrix
$\bm K_z=[K(z_i-z_j)]_{i,j}$ with $K(d)=\int h'(u)h'(u+d)\,du$ is positive
definite. Standard Riemann-sum convergence gives
$\eta_L\bA_L(z)\to \frac14\bm K_z$ entrywise, hence
$c_-L\bI\preceq\bA_L(z)\preceq c_+L\bI$ for all sufficiently large $L$.
Consequently
\[
\mathrm{CRB}_{\rm unlab}(L)=\Theta(\kappa_L/L)
\]
in this well-conditioned interior case. Close records and boundary truncation
violate the conditioning constants; the accounting calculation is unchanged.

\begin{figure*}[t]
\centering
\includegraphics[width=0.60\textwidth]{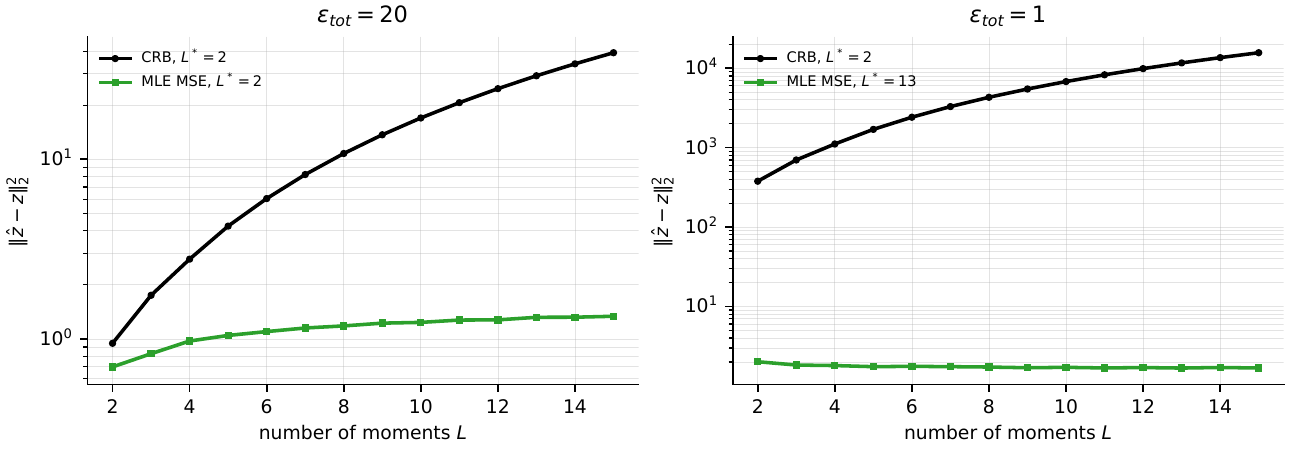}
\caption{Finite-noise fixed-budget release-count sweep in the appendix for the
strict-interior $m=2$ polynomial-moment pair $z=(-0.4,0.7)$, averaged over
$5{,}000$ Gaussian-noise trials per point. At high privacy budgets, the
constrained MLE and local CRB select the same stopping point. At low budgets,
constraints dominate the finite-noise risk, illustrating why the CRB is used
only as a local unbiased benchmark.}
\label{fig:app_mle_vs_L}
\end{figure*}

\phantomsection\label{app:logistic_smoothing}
\paragraph*{Logistic smoothing and numerical details.}
The compact-support assumption in Proposition~\ref{prop:threshold_truncation} can be relaxed for the logistic smoothing used in Section~\ref{sec:experiments}: $h_\tau'(u)=(4\tau)^{-1}\operatorname{sech}^2(u/(2\tau))$ is bounded and has exponential tails, so the same grid/Riemann-sum argument gives $\sum_\ell h_\tau'(z_k-t_\ell)^2=O(L)$ and hence the same $\Omega(\kappa_L/L)$ lower bound. The finite-noise polynomial estimator used in the $m=2$ MLE panels minimizes $(\by-\bp(\bu))^\top\bSigma_L^{-1}(\by-\bp(\bu))$ over $\bu\in[-R,R]^m$, where $\bp(\bu)=(m^{-1}\sum_k u_k^\ell)_{\ell=1}^L$, and then sorts the solution. Runs use the algebraic inversion in \eqref{eq:m2_estimator}, a fixed coarse grid of feasible starts, projected Adam for $60$ iterations with stepsize $0.035$, and the feasible iterate with the best objective value.

\phantomsection\label{app:finite_noise_l_sweep}
\paragraph*{Finite-noise release-count sweep.}
We also ran finite-noise sweeps over $L$ for the $m=2$ polynomial
family. Figure~\ref{fig:app_mle_vs_L} fixes the strict-interior
pair $z=(-0.4,0.7)$ and compares the local CRB with constrained Gaussian
MLE MSE as the number of released moments varies. At
$\varepsilon_{\mathrm{tot}}=20$, where the local approximation is
informative, both curves select $L=2$ in this example. At
$\varepsilon_{\mathrm{tot}}=1$, the constrained estimator is dominated by
box effects and finite-noise bias, so its risk lies far below the local
unbiased CRB and does not validate a CRB stopping rule. This is the
operational boundary emphasized in the main text.

The real-data panels use electricity load-profile values from a real-world grid-data corpus, clipped at the empirical 99th percentile and normalized to $[-1,1]$ before pair/cohort sampling. Additional diagnostics, not shown here, include boundary finite-noise sweeps, finite-noise $L$-sweeps, spectral-bottleneck curves, compact-support threshold controls, accountant/calibration checks, and general-$m$ lower-bound validations.

\end{document}